\documentclass[a4paper,UKenglish,cleveref, autoref, thm-restate,authorcolumns]{lipics-v2019}
\usepackage{xcolor}
\usepackage[section]{placeins}
\usepackage{mathtools}
\usepackage{thm-restate}

\newcommand{\tw}{\operatorname{tw}}
\newcommand{\sier}{\text{Sierp\'inski}}

\newcommand{\kn}{\operatorname{Kn}}
\newcommand{\dsg}{\operatorname{Ds}}

\title{On the treewidth of Hanoi graphs} 

\author{David Eppstein}{University of California, Irvine, United States }{eppstein@uci.edu}{}{}

\author{Daniel Frishberg}{University of California, Irvine, United States }{dfrishbe@uci.edu}{https://orcid.org/0000-0002-1861-5439}{}

\author{William Maxwell}{Oregon State University, United States }{maxwellw@oregonstate.edu}{}{}

\authorrunning{D. Eppstein et al.} 

\Copyright{David Eppstein, Daniel Frishberg, and William Maxwell} 

\ccsdesc[500]{Mathematics of computing~Graph theory} 

\keywords{Hanoi graph, Treewidth, Graph separators, Kneser graph, Vertex expansion, Haven, Tensor product} 

\category{} 

\relatedversion{} 

\supplement{}

\acknowledgements{Some of the results in the section on three-peg Hanoi graphs were previously announced on a web forum~\cite{cstheory}. This research was supported in part by NSF grants CCF-1618301, CCF-1616248, and CCF-1617951.}

\nolinenumbers 

\hideLIPIcs  

\EventEditors{Martin Farach-Colton, Giuseppe Prencipe, and Ryuhei Uehara}
\EventNoEds{3}
\EventLongTitle{10th International Conference on Fun with Algorithms (FUN 2020)}
\EventShortTitle{FUN 2020}
\EventAcronym{FUN}
\EventYear{2020}
\EventDate{September 28--30, 2020}
\EventLocation{Favignana Island, Sicily, Italy}
\EventLogo{}
\SeriesVolume{157}
\ArticleNo{13}

\begin{document}
\maketitle

\begin{abstract}
The objective of the well-known \emph{Towers of Hanoi} puzzle is to move a set of disks one at a time from one of a set of pegs to another, while keeping the disks sorted on each peg. We propose an adversarial variation in which the first player forbids a set of states in the puzzle, and the second player must then convert one randomly-selected state to another without passing through forbidden states. Analyzing this version raises the question of the \emph{treewidth} of \emph{Hanoi graphs}. We find this number exactly for three-peg puzzles and provide nearly-tight asymptotic bounds for larger numbers of pegs.
\end{abstract}

\section{Introduction}
The \emph{Towers of Hanoi} puzzle is very well known (for a comprehensive treatment see \cite{mythsmaths}), but it loses its fun once its player learns the strategy. It has some number $n$ of disks of distinct sizes, each with a central hole allowing it to be stacked on any of three pegs. The disks start all stacked on a single peg, sorted from largest at the bottom to smallest at the top. They must be moved one at a time until they are all on another peg, while at all times keeping the disks in sorted order on each peg. The optimal strategy is easy to follow: alternate between moving the smallest disk to a peg that was not its previous location, and moving another disk (the only one that can be moved). Once one learns how to do this, and that the strategy takes $2^n-1$ moves to execute~\cite{3pegalg}, it becomes tedious rather than fun.

The puzzle can be modified in several ways to make it more of an intellectual challenge and less of an exercise in not losing one's place. One of the most commonly studied variations involves using some number $p$ of pegs that may be larger than three. Of course, one can ignore the extra pegs, but using them allows shorter solutions. An optimal solution for four pegs was given by Bousch in 2014~\cite{Bousch2014}, but the best solution for larger numbers of pegs remains open.
The Frame--Stewart algorithm solves these cases, but it is not known if it is optimal~\cite{FrameStewart}. The length of an optimal solution, for starting and ending positions of the disks chosen to make this solution as long as possible, can be modeled graph-theoretically using a graph called the \emph{Hanoi graph}, which we denote $H_p^n$.
This graph is formed by constructing a vertex for each configuration of the game, and connecting two vertices with an edge when their configurations are connected by one legal move. The number of moves between the two farthest-apart positions is then the diameter of this graph. For three pegs, the diameter of $H_3^n$ is $2^n-1$ (the traditional starting and ending positions are the farthest apart) but for $p>3$ the diameter of $H_p^n$ is unknown~\cite{Cartesian}.

In this paper, we consider a different way of making the puzzle more difficult, by making it adversarial. In our version of the game, the first of two players selects a predetermined number of \emph{forbidden positions}, that the second player cannot pass through. Then, the second player must solve a puzzle using the remaining positions. If that were all, then the first player could win by forbidding only a very small number of positions, the $p-1$ positions one move away from the start position. To make the first player work harder, after the first player chooses the forbidden positions, we choose the start and end position randomly from among the positions in the game. We ask: How many positions must the first player forbid, in order to make this a fair game, one where both players have equal chances of being able to win?

We can model this problem graph-theoretically, as asking for the smallest number of vertices to remove from a Hanoi graph in order for the number of pairs of remaining vertices belonging within the same component as each other to be half the total number of pairs of vertices. The answer to the problem lies between the minimum size of a \emph{balanced vertex separator} (\autoref{game-from-separator}) and (up to a constant factor of three) the minimum order of a \emph{recursive balanced vertex separator}; the latter is equivalent, up to constant factors, to asking for the \emph{treewidth} of $H_p^n$.
(Technically, the treewidth can be larger than the recursive separator order by a logarithmic factor when this order is constant, but both are within constant factors of each other when the order is polynomial.)
Treewidth is of interest to computer scientists as many NP-hard graph problems become fixed-parameter tractable on graphs with bounded treewidth \cite{FPTsurvey}.

\subsection{New results and prior work}
We conjecture that the treewidth of $H_p^n$ is $\Theta((p-2)^n)$.
For $p>3$ this bound is exponential, and we make progress towards this conjecture by proving that the treewidth is within a polynomial factor of this bound. More precisely we show an asymptotic upper bound of $O((p-2)^n)$ and an asymptotic lower bound of $\Omega(n^{-(p-1)/2} \cdot (p-2)^n)$.
We increase the lower bound to $\Omega(\frac{2^n}{n})$ when $p = 4$.
Moreover, we find the exact (constant) treewidth of $H_3^n$ and of the closely-related $\sier{}$ graphs. Our results provide an answer to our motivating question on sizes of forbidden sets of positions, up to polynomial factors for four or more pegs and exactly for three pegs.

As a byproduct of our proof techniques, we observe a nearly linear asymptotic lower bound on the treewidth of the \emph{Kneser graph} (Corollary~\ref{cor:kntw}). Harvey and Wood~\cite{Harvey2014} showed a previous exact result for the treewidth of $\kn(n, k)$ when $n$ is at least quadratic in $k$.
Another byproduct of our proof techniques gives a new lower bound on the treewidth of the \emph{tensor product} $G \times H$ of two graphs $G$ and $H$, when $H$ is not bipartite. Eppstein and Havvaei \cite{eppstein_et_al:LIPIcs:2019:10217} gave an upper bound on the treewidth of $G \times H$; Brev\v{s}ar and Spacapan \cite{brevspac} gave an analogous lower bound for edge connectivity; Kozawa et al. \cite{kozawa} gave lower bounds for the treewidth of the strong product and Cartesian product of graphs.

\section{Preliminaries}
\subsection{Hanoi graphs}
Label the $n$ disks of the Towers of Hanoi, in order of increasing size, as  $d_1, \dots, d_n$. If disks $d_i$ and $d_j$ are on the same peg, and $i < j$, then $d_j$ is constrained to be below $d_i$. A legal move in the game consists of moving the top (smallest) disk on some peg $A$ to another peg $B$, while preserving the constraint. At the beginning of the game, all $n$ disks are on the first peg. The objective of the game is to obtain, through some sequence of legal moves, a state in which all $n$ disks are on the last peg. Let $p$ be the number of pegs. Traditionally, $p = 3$.

Formally, a \emph{configuration} of the $p$-peg, $n$-disk Towers of Hanoi game is an $n$-tuple $(p_1, p_2, \dots, p_n)$ where $p_i \in \{1, 2, \dots, p\}$, describing the peg for each disk $d_i$.
We say two configurations $(p_1, p_2, \dots, p_n)$ and $(p'_1, p'_2, \dots, p'_n)$ are \emph{compatible} if a move from one configuration to the other is allowed. This happens exactly when the two configurations differ only in the value of a single coefficient $p_i$, for which $d_i$ is the smallest disk having either of the two differing values.
We call a configuration with each disk on the same peg a \emph{perfect state}.
The \emph{Hanoi graph} $H_p^n$ is a graph whose vertices are the configurations of the $n$-disk, $p$-peg Towers of Hanoi game, with an edge for each compatible pair of configurations. It has $p^n$ vertices and $\frac{1}{2} {p \choose 2} (p^n - (p - 2)^n)$ edges~\cite{CountingHanoi}.

\subsection{Recursive balanced separators, treewidth, and havens}
In this section we give a brief discussion of the concepts of recursive balanced separators, treewidth, and havens. 
Given a graph $G=(V,E)$ a \emph{vertex separator} is a subset $X \subseteq V$ such that $G \setminus V$ consists of two disjoint sets of vertices $A$ and $B$ with $A \cup B = V \setminus X$ and for all $a \in A$, $b \in B$ there is no edge $(a, b)$ in the graph $G \setminus X$.
Further, given a constant $c$ with $\frac{1}{2} \leq c < 1$, we call $X$ a \emph{balanced vertex separator} if $(1-c)|V| \leq |A| \leq \frac{|V|}{2}$ and $\frac{|V|}{2} \leq |B| \leq c|V|$. When this holds we call $X$ a \emph{c-separator}.
We say that $G$ has a \emph{recursive balanced separator} of order $s$, where $s:\mathbb{N}\rightarrow\mathbb{N}$ is a nondecreasing function, whenever either $|V| \leq 1$, or we can find a balanced separator of size $s(|V|)$ for $G$, and the resulting subgraphs $A$ and $B$ have recursive balanced separators of order $s$ respectively. We abuse notation and refer to $s(|V|)$ as $s(G)$. 

A \emph{tree decomposition} of a graph $G$ is a tree $T$ whose nodes are sets of vertices in $G$ called \emph{bags}, such that the following conditions hold.
\begin{itemize}
\item If two vertices are adjacent, then they share at least one bag.
\item If a vertex $v$ is in two bags $A$ and $B$, then $v$ is in every bag on the path from $A$ to $B$ in $T$.
\item Every vertex in $V(G)$ is in some bag.
\end{itemize}
The \emph{width} of a tree decomposition $T$ is one less than the maximum size of a bag in $T$. The \emph{treewidth} of a graph $G$, denoted $\tw(G)$, is the minimum width over all tree decompositions of $G$.
The bags in the tree decomposition $T$ induce vertex separators in $G$. Moreover, we can use the tree decomposition to find a recursive balanced separator for $G$.
Hence, the treewidth of $G$ is a measure of the minimum order of a recursive balanced separator for $G$.
The following folklore lemma relates the order of a recursive balanced separator to the treewidth of a graph; see \cite{erickson} and \cite[Lemma 6.6]{NesOss-S-12}.

\begin{lemma}\label{lem:twsep}
Let $G$ be an $N$-vertex graph. If $t = \tw(G)$, then with respect to every constant $\frac{1}{2} \leq c < 1$, $G$ has a recursive balanced separator of order $s(N') = t + 1$ for all $1 \leq N' \leq N$. On the other hand, if $G$ has recursive balanced separator of order $t$, where $t = \Omega(N^d)$ for some constant $d > 0$, then $G$ has treewidth $O(t)$.
\end{lemma}

Returning to our motivating game, in which one player forbids the use of a designated set of states in the state space of a puzzle and the other player attempts to connect two randomly chosen states by a path, we see that a fair number of states to forbid is controlled by the size of a recursive balanced separator.
We formalize this in the following lemma:

\begin{lemma}
\label{game-from-separator}
Given a graph $G$, let $f(G)$ be the minimum number of vertices that can be removed from $f$ so
that, if two random vertices of $G$ are chosen, the probability that they are not removed and have a path between them is at most $1/2$. Let $c = 1/\sqrt{2}$, and let $r(G)$ be the minimum size of a $c$-separator (not necessarily recursive) for $G$. Let $s$ be the minimum order of a recursive $c$-separator for $G$. Then $r(G)\le f(G)\le 3s(G)$.
\end{lemma}

\begin{proof}
If we remove a vertex set $X$ with $|X| = f(G)$, leaving probability less than $1/2$ that two randomly-chosen vertices from $G$ are connected, then the remaining subgraph cannot contain any connected component larger than $|V(G)|/\sqrt{2}$. If it contains any connected component of size at least $|V(G)|/2$, then $f$ separates that subgraph from the remaining vertices, and otherwise the remaining small subgraphs can be combined to give a separation between two subgraphs whose largest size is at most $2|V(G)|/3$, better than $c$. Therefore, $r(G)\le f(G)$.

To show that $f(G)\le 3s(G)$, find a recursive $c$-separator for $G$ of order $s$; the separator $X$ has the following three separators as subsets: a $c$-separator $X$ for $G$ resulting in two separated subgraphs, and $c$-separators $Y$ and $Z$ for each of the two separated subgraphs. $|X| + |Y| + |Z| \leq 3s(G)$. Removing $X \cup Y \cup Z$ from $G$ partitions the rest of $G$ into subgraphs of size at most $|V(G)|/2$. No matter which of these subgraphs one of the randomly chosen two vertices belongs to, the probability that the other vertex belongs to the same component will be at most $1/2$.
\end{proof}

Some of our results will bound the treewidth of graphs using \emph{havens}, a mathematical formalization of an escape strategy for a robber in cop-and-robber pursuit-evasion games.
In these games, a set of cops and a single robber are moving around on a given graph $G$.
Initially the robber is placed at any vertex of the graphs, and none of the cops has been placed.
In any move of the game, one of the cops can be removed from the graph, or a cop that has already been removed can be placed on any vertex of the graph. However, before the cop is placed,
the robber (knowing where the cop will be placed) is allowed to move along any path in the graph that is free of other cops. The goal of the cops is to place a cop on the same vertex as the robber while simultaneously blocking all escape routes from that vertex, and the goal of the robber is to evade the cops forever. In these games, a \emph{haven of order} $k$ describes a strategy by which the robber can perpetually evade $k$ cops, by specifying where the robber should move for each possible move by the cops. It is defined as a function $\phi$, mapping each set of vertices $X \subseteq V$ with $|X| \leq k$ to a nonempty connected component in $G \setminus X$, such that whenever $X_1 \subseteq X_2$, $\phi(X_2) \subseteq \phi(X_1)$.
A robber following this strategy will move to any vertex of $\phi(X)$, where $X$ denotes the set of vertices to be occupied by the cops at the end of the move. The mathematical properties of havens ensure that the robber can always reach one of these vertices by a cop-free path.

Returning again to our adversarial version of the Towers of Hanoi puzzle, the cops-and-robber game is equivalent to a game in which the first player attempts to pin the second player to a state from which no legal move to any non-forbidden state is possible. The placement (or removal) of a cop is equivalent to the first player designating (or de-designating) a state as forbidden; an evasion strategy for a robber is equivalent to the existence of a legal move for the second player.

The existence of a haven in $G$ yields a lower bound on the treewidth of $G$ via the following lemma.

\begin{lemma}[Seymour and Thomas \cite{Seymour1993}]
A graph $G$ has a haven of order greater than or equal to $k$ if and only if $\tw(G) \geq k - 1$.
\label{lem:twhaven}
\end{lemma}

\section{Three pegs}
In this section we show that $\tw(H_3^n) \leq 4$ for all $n \geq 1$.
We prove this by relating the three-peg Towers of Hanoi game and the \sier{} triangle graphs, which we denote $S_n$. 
$S_n$ has treewidth at least 3 for all $n$, as it contains a triangle, and (\autoref{lem:siertw}) it equals 4 for $n>4$.
Additionally, each \sier{} triangle graph contains a smaller three-peg Hanoi graph as a minor, and vice versa.
From this it will follow that $\tw(H_3^n) = 4$ for all sufficiently large $n$.
For completeness we include a more detailed proof of the bounds on $\tw(S_n)$.

We define the \sier{} triangle graphs, along with a planar embedding of them, inductively. 
The planar embedding will allow us to see the geometric similarity between the \sier{} graphs and the three-peg Hanoi graphs.
The first \sier{} triangle, $S_1$, is isomorphic to $K_3$ with a planar embedding of an equilateral triangle with unit length sides. The vertices of the triangle coincide with the vertices of $K_3$.

Inductively, we assume $S_{n-1}$ has a planar embedding whose outer face is embedded geometrically as an equilateral triangle.
We label the vertices on the outer face of the triangle $v_\ell, v_r, v_t$ which are the left, right, and top vertices, respectively.
To construct $S_n$ from $S_{n-1}$ we take three copies of $S_{n-1}$ labeled $S_{n-1}^L, S_{n-1}^R, S_{n-1}^T$ for the left, right, and top triangles and make the following vertex identifications.
\begin{enumerate}
\item Identify $v_\ell$ in $S_{n-1}^R$ with $v_r$ in $S_{n-1}^L$, and call the resulting vertex $v_{\ell r}$.
\item Identify $v_t$ in $S_{n-1}^L$ with $v_\ell$ in $S_{n-1}^T$, and call the resulting vertex $v_{\ell t}$.
\item Identify $v_t$ in $S_{n-1}^R$ with $v_r$ in $S_{n-1}^T$, and call the resulting vertex $v_{rt}$.
\end{enumerate}

The resulting graph has a planar embedding whose outer face can again be embedded as a subdivided equilateral triangle. In $S_n$ the left, right, and top vertices of the outer face are contained in $S_{n-1}^L, S_{n-1}^R$, and $S_{n-1}^T$ respectively. As before we denote them as $v_\ell, v_r,$ and $v_t$.
Note that we can recursively decompose $S_n$ into a triangle and a trapezoid, from which the trapezoid further decomposes into two additional triangles. (Here, we only consider trapezoids whose long side is horizontal.) This recursive decomposition leads to the construction of a tree decomposition of $S_n$.
The six distinguished vertices $v_\ell, v_r, v_t, v_{\ell r}, v_{\ell t},$ and $v_{rt}$ define the bags of the tree decomposition at each level.
The set $\{v_t, v_{\ell t}, v_{r t}, v_\ell, v_r\}$ lies on the perimeter of a triangle in this decomposition. We call a bag in the tree decomposition consisting of these vertices a \emph{triangular bag}.
On the other hand, the set $\{v_{\ell t}, v_{rt}, v_\ell, v_{\ell r}, v_r\}$ lies on the perimeter of a trapezoid in the decomposition.
We call a bag in the tree decomposition consisting of these vertices a \emph{trapezoidal bag}.
With this definition we are now ready give a proof of the fact that $\tw(S_n) = 4$ for all $n > 4$. 
\begin{figure}[!htb]
    \centering
    \includegraphics[scale=0.25]{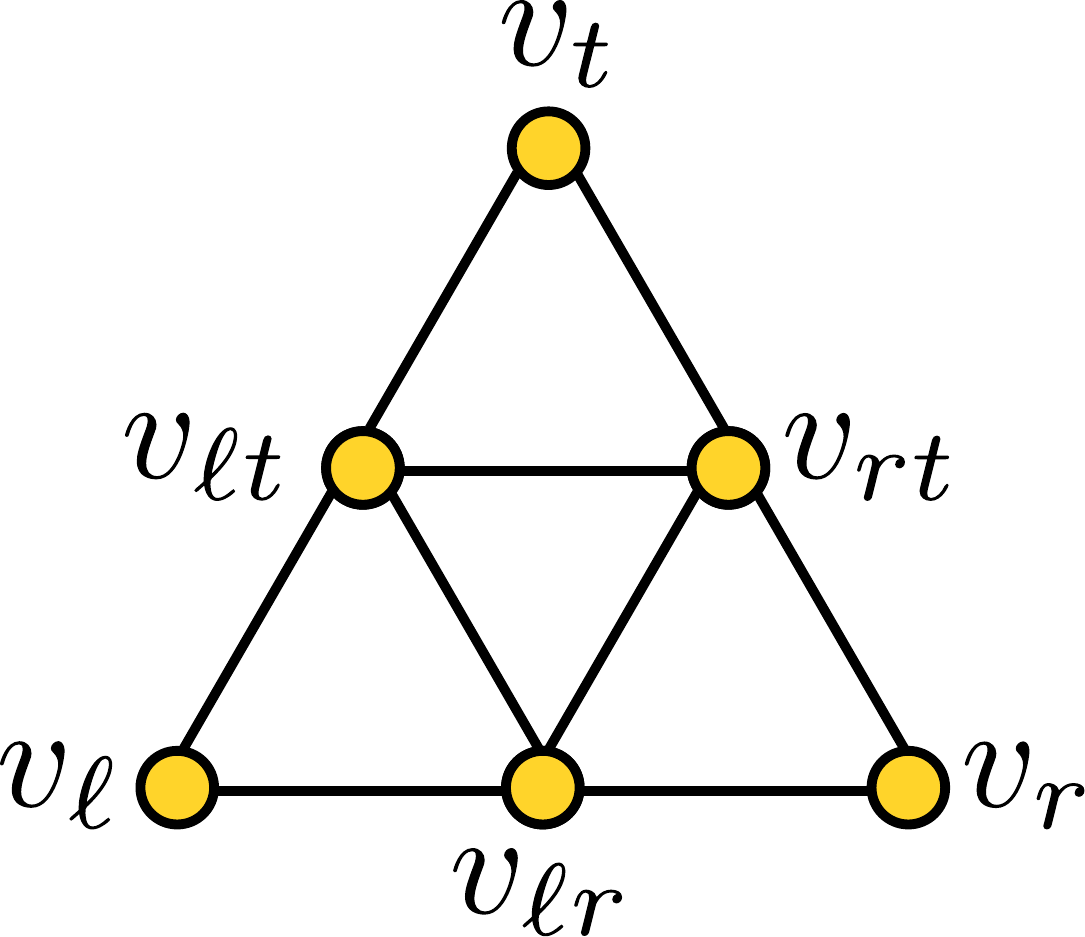}\quad\quad
    \includegraphics[scale=0.25]{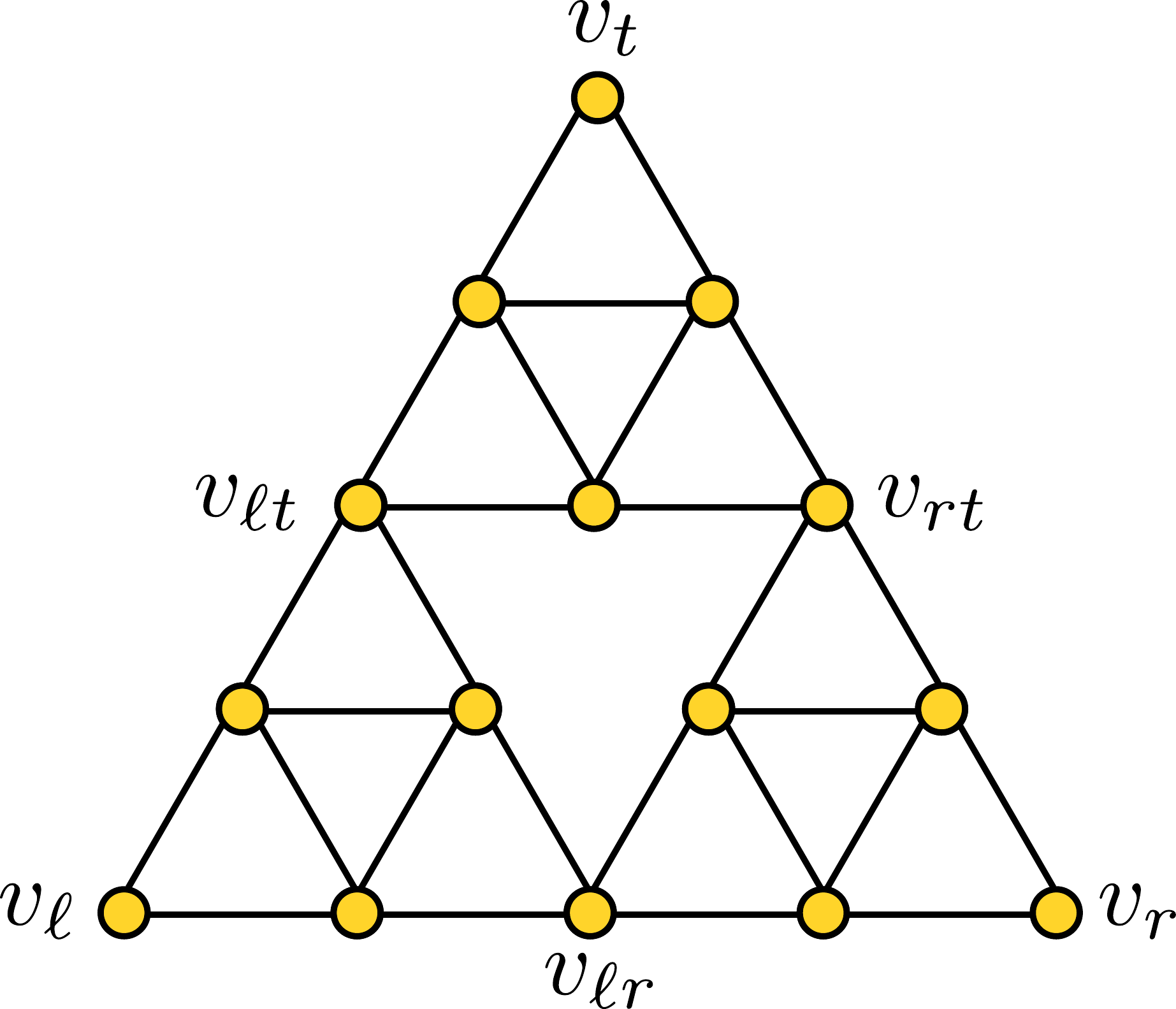}
    \caption{The \sier{} graphs $S_2$ and $S_3$.}
    \label{fig:sier}
\end{figure}

\begin{lemma}
\label{lem:siertw}
The treewidth of $S_n$ is equal to $4$ for all $n > 4$.
\end{lemma}
\begin{proof}
To prove the upper bound we construct a tree decomposition of $S_n$ out of the triangular and trapezoidal bags defined above.
We take the triangular bag in $S_n$ to be the root of the tree decomposition, and recursively decompose $S_n$ into its triangular and trapezoidal subgraphs.
A bag at depth $k$ is either a triangular or trapezoidal bag from an $S_{n-k}$ subgraph.
The children of a trapezoidal bag at depth $k$ are the triangular bags corresponding to the two copies of $S_{k-1}$ that make up the trapezoid.
The children of a triangular bag at depth $k$ are a trapezoidal and a triangular bag corresponding to the decomposition of $S_k$ into a trapezoid and triangle.
Every edge of $S_n$ is contained in some triangle or trapezoid, and every triangle and trapezoid appear as a bag in the tree decomposition.
For any vertex $v$ in $S_n$ if $v \in B_1, B_2$ where $B_1$ and $B_2$ are distinct bags there are two cases to consider.
If $B_1$ is an ancestor of $B_2$ then $v$, by the construction of the bags, must be in every triangular or trapezoidal bag lying in between them.
If there is no ancestry relationship, then $v$ must lie in the intersection of the shapes defined by $B_1$ and $B_2$. Hence, there is some triangle or trapezoid containing both $B_1$ and $B_2$ which is their least common ancestor in the tree decomposition.
See Figure~\ref{fig:tree_decomp} for an illustration on $S_3$.

To prove the lower bound it is sufficient to show that $S_n$ contains a subdivision of the octahedron graph when $n > 4$. The octahedron graph is a forbidden minor for treewidth 3 graphs \cite{Arnborg1990}. See Figure \ref{fig:octa} for an illustration.
\end{proof}

\begin{figure}[!htb]
    \centering
    \includegraphics[scale=0.2]{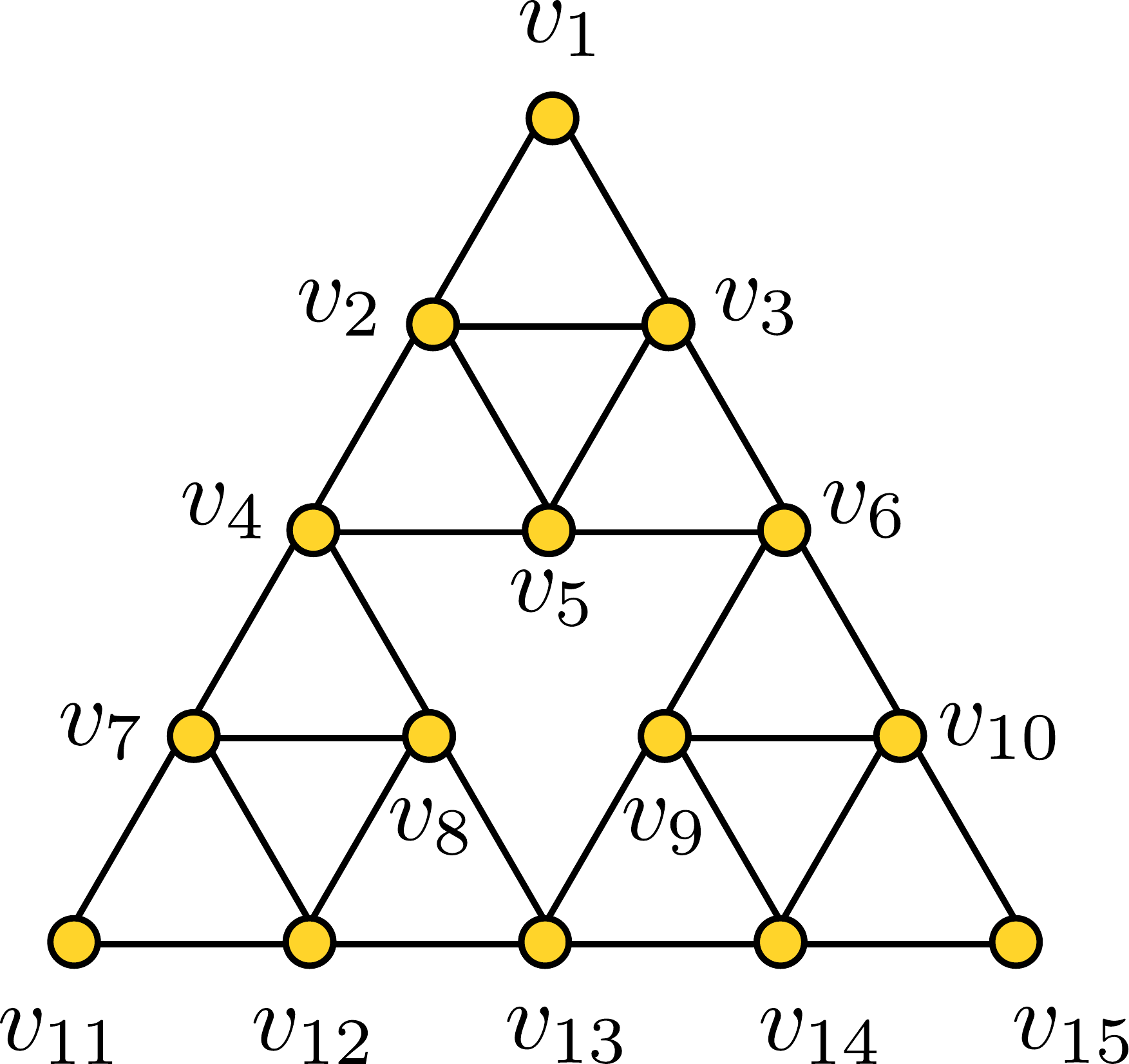}\quad\quad\quad
    \includegraphics[scale=0.14]{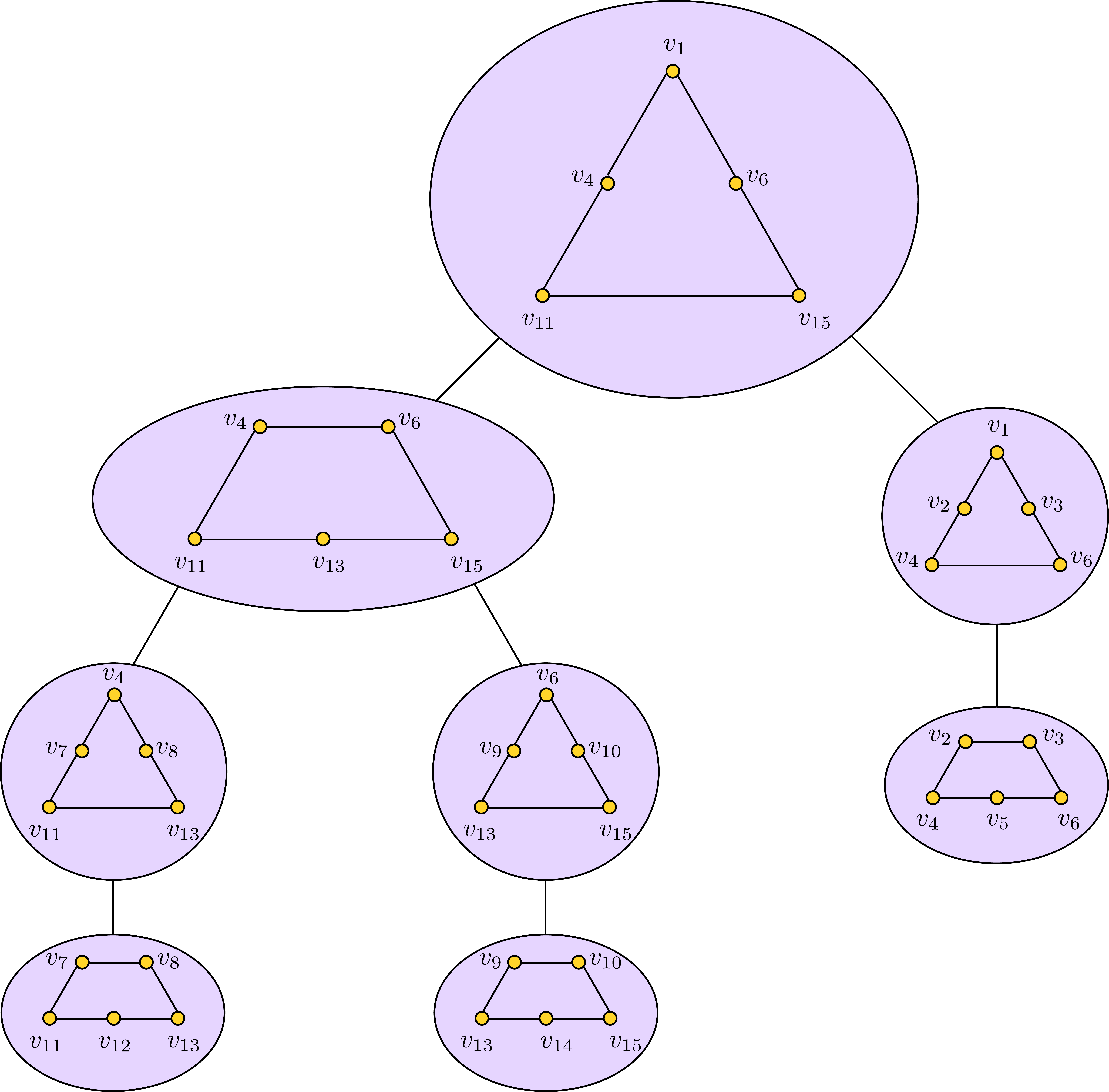}
    \caption{$S_3$ along with its tree decomposition.}
    \label{fig:tree_decomp}
\end{figure}

\begin{figure}[!htb]
    \centering
    \includegraphics[scale=0.33]{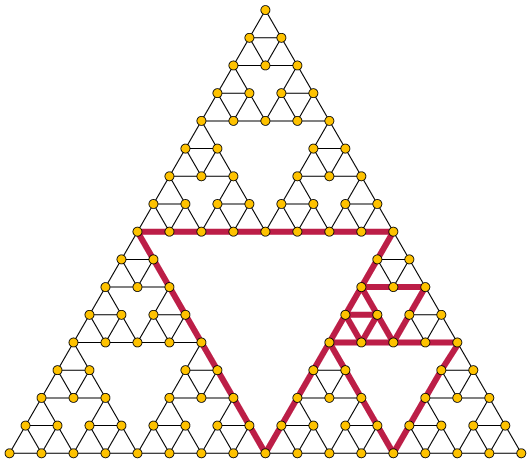}
    \caption{The graph $S_5$ with a subdivision of the octahedral graph highlighted in red.}
    \label{fig:octa}
\end{figure}

Next we give an inductive construction of the Hanoi graph $H_3^n$ with $3$ pegs and $n$ disks. This construction is almost identical to that of $S_n$, but instead of identifying vertices we connect the three copies of $H_3^{n-1}$ with three edges.
Recall that the vertices of $H_3^n$ are configurations representing the game state, that is a vertex is an element of $\{1, 2, 3\}^n$.
We define $H_3^1$ to be $K_3$ with the same planar embedding as in the case of the \sier{} triangle and denote the vertices as the $1$-tuples $(1), (2), (3)$. The cyclic ordering of the vertices does not affect our construction.

By induction we assume $H_3^{n-1}$ has a planar embedding whose outer face is an equilateral triangle such that the corners of the triangle are the configurations corresponding with the perfect states, and we denote these vertices $p_1,p_2,p_3$. For $i \in \{1,2,3\}$ let $H_i$ be the graph isomorphic to $H_3^{n-1}$ with the vertex set $V(H_3^{n-1}) \times \{i\}$.
We construct $H_3^n$ out of the three $H_i$'s and add the following edges.

\begin{enumerate}
\item Add an edge between $(p_1, 2)$ and $(p_1, 3)$ and denote it $e_{\ell r}$.
\item Add an edge between $(p_2, 1)$ and $(p_2, 3)$ and denote it $e_{r t}$.
\item Add an edge between $(p_3, 1)$ and $(p_3, 2)$ and denote it $e_{\ell t}$.
\end{enumerate}

We call these three edges the \emph{boundary edges}.
The boundary edges represent the legal moves obtained by moving the largest disk. It is clear from the construction that the resulting graph embeds into the plane as an equilateral triangle with the perfect states at the corners of the triangle.
See Figure \ref{fig:hanoi}.

\begin{figure}[h]
    \centering
    \includegraphics[scale=0.25]{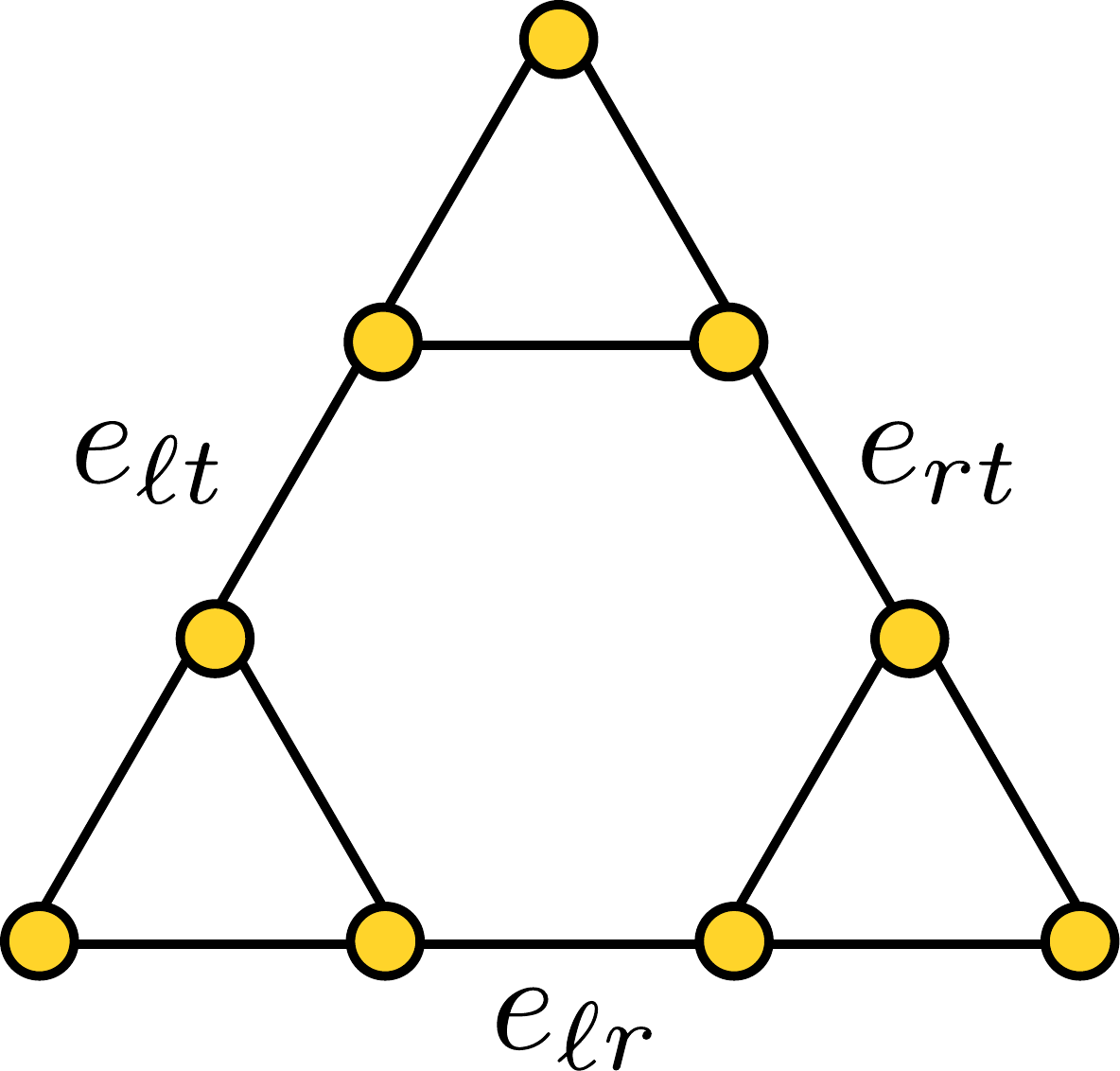}\quad\quad
    \includegraphics[scale=0.25]{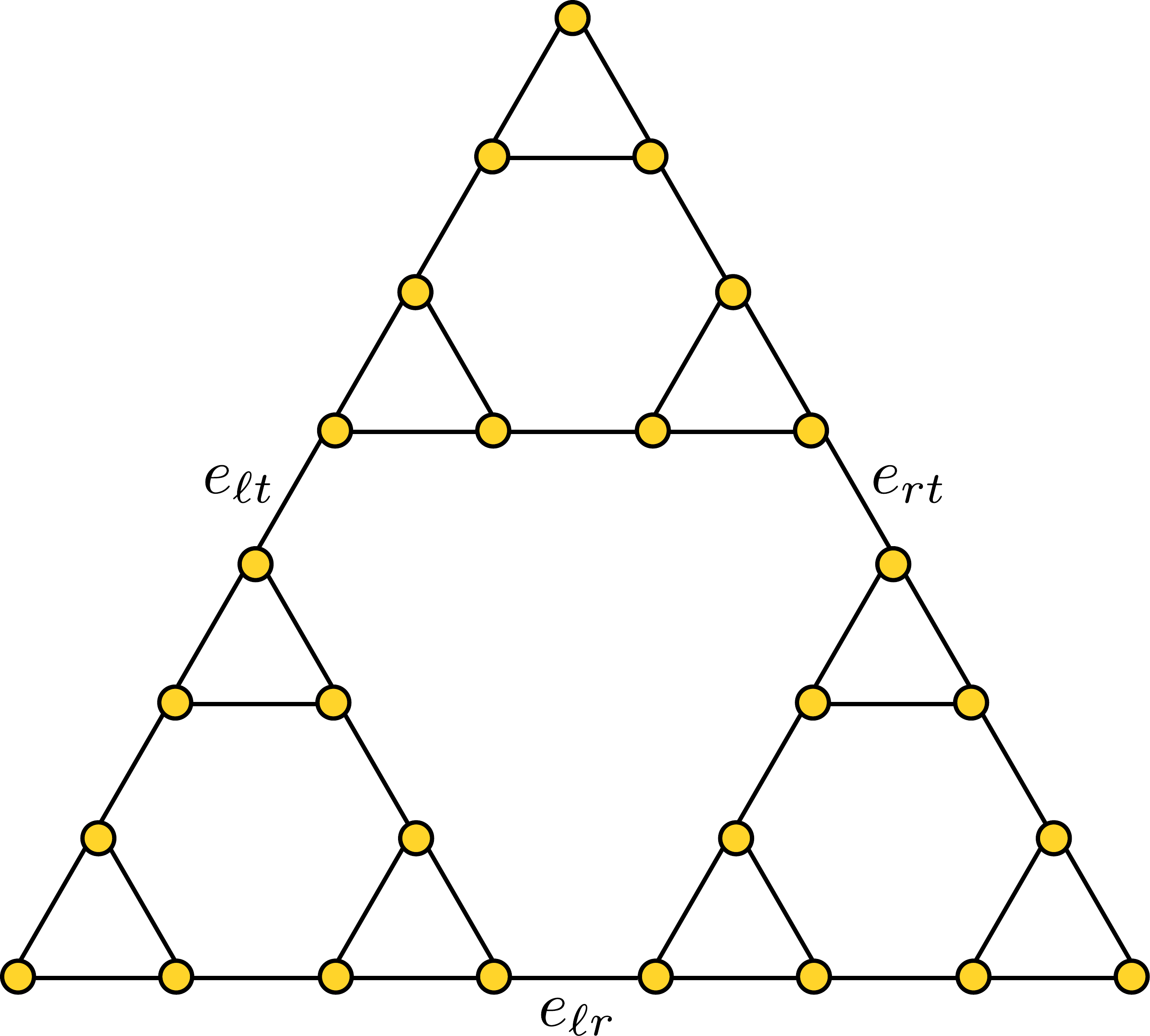}
    \caption{The Hanoi graphs $H_3^2$ and $H_3^3$. We label the boundary edges such that their index coincides with their corresponding vertex in the \sier{} triangle.}
    \label{fig:hanoi}
\end{figure}

\begin{theorem}\label{thm:twh3}
$\tw(H_3^n) = 4$ for all $n > 4$.
\end{theorem}
\begin{proof}
To prove the lower bound we contract the boundary edges of $H_3^n$ to create an $S_n$-minor. Hence, $4 = \tw(S_n) \leq \tw(H_3^n)$ for $n > 4$.

To get the inequality $\tw(H_3^n) \leq 4$ we inductively construct an $H_3^{n}$-minor of $S_{n+1}$ as follows.
For the base case we can easily find a copy of $K_3$ in $S_2$.
Let $G_1, G_2, G_3$ be the three $S_{n}$ subgraphs used to construct $S_{n+1}$ and let $v_{i,j}$ be the vertex shared by $G_i$ and $G_j$.
By the inductive hypothesis we assume each $G_i$ contains an $H_3^{n-1}$-minor which we denote by $H_i$. We construct an $H_3^{n}$-minor in $S_{n+1}$ by connecting the corresponding perfect states of $H_i$ and $H_j$ via a path containing $v_{i,j}$ for each $i \neq j$.
These paths can be chosen to be vertex-disjoint, which proves the theorem. See Figure \ref{fig:hanoiminor} for an illustration.
\end{proof}

\begin{figure}[!htb]
    \centering
    \includegraphics[scale=0.25]{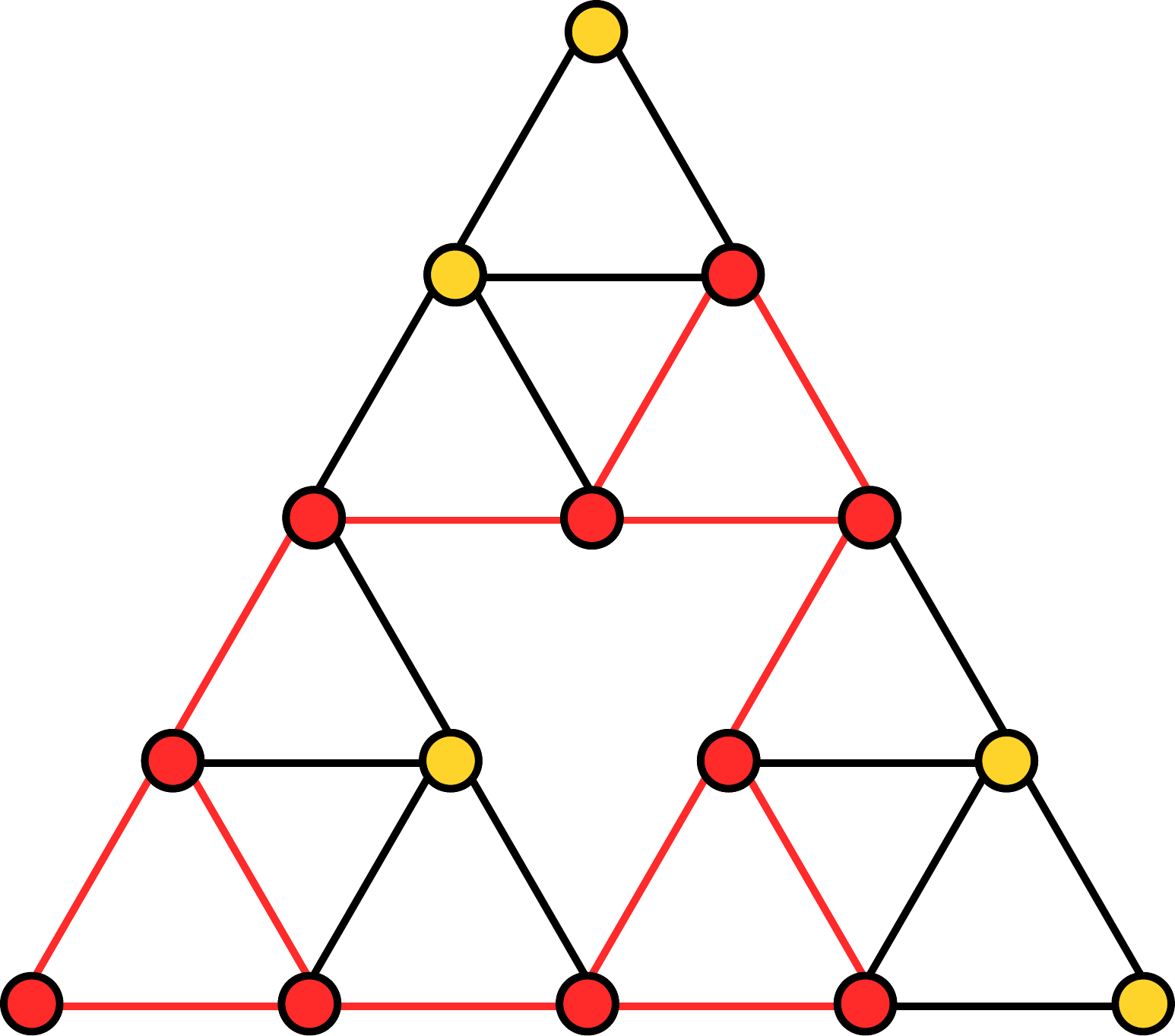}
    \caption{$S_3$ with an $H_3^2$ minor highlighted in red.}
    \label{fig:hanoiminor}
\end{figure}

The three-peg case is simple enough that we can analyze our forbidden-state version of the puzzle directly. If two states are forbidden, the only way to separate the remaining states is to separate one recursive subgraph of the same type from the rest of the graph. In terms of the original puzzle,
the two forbidden states can be described by choosing a peg and a number $k$ and forbidding the two states where the largest $k$ disks are on the chosen peg and the remaining $n-k$ disks are all on the same peg as each other (one of the other two pegs). The probability of a connection between two randomly-chosen states is maximized for $k=1$, for which, for large $n$, the probability of a path between two randomly-chosen vertices becomes approximately $(2/3)^2+(1/3)^2=5/9$.
On the other hand, if three states are forbidden, it becomes possible to separate the state space into three equally-sized subgraphs by forbidding three of the six states in which the largest disk can move. For this selection, the probability of a path between two randomly-chosen vertices becomes
$3(1/3)^2=1/3$.

\section{More pegs}
We conjecture that the treewidth of the Hanoi graph $H_p^n$ is $\Theta((p-2)^n)$.
By \autoref{lem:twsep} the same bound would automatically apply to the recursive balanced separator orders of these graphs; by \autoref{game-from-separator}, this would imply an upper bound on the number of states to forbid to make the adversarial version of the Hanoi puzzle fair ($f(G)$ in \autoref{game-from-separator}).
In this section we make progress towards this conjecture by proving the asymptotic upper bound $\tw(H_p^n) = O((p-2)^n)$ and the asymptotic lower bound $\tw(H_p^n) = \Omega (n^{-(p-1)/2} \cdot (p - 2)^n)$. We obtain the lower bound by proving that every balanced separator of $H_p^n$ (recursive or otherwise) is of this asymptotic order. This lower bound then applies to $f(G)$ in \autoref{game-from-separator}.
Our bounds are almost tight, off only by the factor $\Theta(n^{(p-1)/2})$.
We begin by proving the asymptotic upper bound, which we do by constructing a recursive balanced separator of the required order and applying Lemma~\ref{lem:twsep}.

\begin{theorem}
For any fixed $p \geq 3$ and $n \geq 1$, $\tw(H_p^n) = O((p-2)^n)$.
\label{thm:ubound}
\end{theorem}
\begin{proof}
We can recursively decompose $H_p^n$ into $p$ vertex-disjoint copies of $H_p^{n-1}$ by considering the subgraphs induced by fixing the position of the largest disk in the configurations.
We call a vertex a boundary vertex if in its configuration there is at least one peg occupied by no disks and the largest disk shares its peg with no other disks. These are the configurations in which the largest disk is free to move. The endpoints of edges between the $H_p^{n-1}$ subgraphs in our decomposition are the boundary vertices.

We compute the order of our recursive balanced separator by counting the number of boundary vertices. This is the number of ways to distribute $n-1$ disks across $p-2$ pegs, hence the size of the separator is ${p \choose 2}(p-2)^{n-1}$. Our choice of separator splits $H_p^n$ into $p$ subgraphs of size $\frac{1}{p}|V(H_p^n)|$. By grouping the $H_p^{n-1}$ subgraphs into two vertex sets, we obtain a $c$-separator where $c \in \{\frac{\lceil p/2 \rceil}{p}, \frac{\lceil p/2 \rceil + 1}{p}, \dots, \frac{p-1}{p}\}$ depending on our choice of vertex sets. Each subgraph can then be recursively decomposed in a similar way, and the number of vertices required in each recursive decomposition at level $i$ is equal to ${p \choose 2}(p - 2)^{n-i}$. The theorem follows directly from Lemma~\ref{lem:twsep}.
\end{proof}

To prove the asymptotic lower bound we construct a new graph related to $H_p^n$ whose treewidth is easier to compute.
We can specify the positions of a subset of disks in a Hanoi puzzle by a mapping $\rho \colon [n] \rightarrow [p] \cup \{\infty\}$, where a finite value of $\rho(i)$ specifies the peg containing disk $d_i$ and an infinite value means that disk $d_i$ is allowed to be placed on any peg that does not also contain a specified disk. We define the \emph{pegset} induced by $\rho$ to be the states consistent with this specification.
More formally, a vertex $v = (p_1, p_2, \dots, p_n)$ is in the pegset induced by $\rho$ if and only if :

\begin{enumerate}
\item for all $k\in [n]$, if $\rho(k) \neq \infty$ then $\rho(k) = p_k$, and
\item for all $k, l \in [n]$, if $\rho(k) = \infty \ne \rho(l)$, then $p_k\ne p_l$.
\end{enumerate}

If $\rho(k) \neq \infty$ we call $d_k$ \emph{frozen} by $\rho$; further, if a peg $p_k$ is in the image of $\rho$ we call $p_k$ frozen by $\rho$ as well.
Intuitively, a pegset is the result of freezing a set of disks onto a set of pegs and playing a Hanoi puzzle using only the remaining unfrozen disks and pegs.
We are interested in pegsets that meet two additional properties:

\begin{enumerate}
\setcounter{enumi}{2}
\item exactly $p - 3$ elements of $[p]$ have a non-empty inverse under $\rho$, and
\item for $j \in [p]$ either $|\rho^{-1}(j)| = \lfloor\frac{n-1}{p-2}\rfloor$ or $|\rho^{-1}(j)| = 0$.
\end{enumerate}

\noindent
We call such pegsets \emph{regular pegsets}. Note that, because we still have three pegs unfrozen, and because the three-peg Hanoi graphs are connected, each regular pegset describes a connected subgraph of the Hanoi graph.

To make our analysis cleaner we assume that $n \equiv 1 \mod (p - 2)$, hence properties 3 and 4 imply that there are precisely $\frac{n-1}{p-2} + 1$ unfrozen disks in a regular pegset.
Note that this restriction on $n$ does not change the overall asymptotic analysis for other values of $n$, as we can still lower-bound the treewidth for other $n$ by rounding $n$ down to a value with this restricted form.

Let $I_p^n$ denote the graph whose vertices are the regular pegsets of $H_p^n$ where two vertices share an edge if and only if the intersection of their corresponding pegsets is non-empty. We call $I_p^n$ the \emph{pegset intersection graph} of $H_p^n$.
We characterize the adjacency condition in terms of frozen disks and pegs in Lemma~\ref{lem:adj}.

\begin{lemma}\label{lem:adj}
Two regular pegsets $u$ and $v$ are adjacent in $I_p^n$ if and only if the following criteria are satisfied:
\begin{enumerate}
\item if a disk is frozen by both $u$ and $v$, then both $u$ and $v$ freeze it to the same peg,
\item $u$ and $v$ each freeze exactly one peg unfrozen by the other,
\item if a disk is frozen by $u$ but not by $v$, then the peg it is frozen on is not frozen by $v$, and
\item if a disk is frozen by $v$ but not by $u$, then the peg it is frozen on is not frozen by $u$.
\end{enumerate}
\end{lemma}
\begin{proof}
If $u$ and $v$ are adjacent in $I_p^n$ then there exists some vertex $w = (w_1, w_2, \dots, w_n)$ contained in both pegsets.
By the definition of a regular pegset both $u$ and $v$ freeze $\frac{n-1}{p-2}$ disks evenly across $p-3$ pegs and leave $\frac{n-1}{p-2}+1$ disks unfrozen.
We will now show that each of the four claims follows from the adjacency of $u$ and $v$.

1. Suppose for a contradiction that a disk $d_i$ is frozen to different pegs by $u$ and $v$; then no configuration in $u$ can equal a configuration in $v$ since they differ at the $i$th component.

2. If $u$ and $v$ freeze the same set of pegs then a configuration in $u$ cannot equal a configuration in $v$ since they will differ on the components corresponding to frozen disks.
Now, assume $u$ freezes more than one peg left unfrozen by $v$. A vertex $w$ in the intersection of $u$ and $v$ would have a configuration that matches both $u$ and $v$ on their frozen disks, but the total number of disks frozen by $u$ and $v$ is at least $(p-1) \cdot \frac{n-1}{p-2}$; then $w$ has more than $n$ disks, contradicting the fact that $w$ is a valid configuration.

3. Let $u$ freeze the disk $d_i$ onto the peg $p_k$ and assume $v$ does not freeze $d_i$.
If $v$ also freezes $p_k$ then $v$ must freeze $\frac{n-1}{p-2}$ disks onto $p_k$ while leaving $d_i$ unfrozen, hence there is no configuration in $v$ that places $d_i$ onto $p_k$.

4. Identical to 3.

We are now ready to prove the converse. Let $u$ and $v$ be pegsets in $I_p^n$ such that conditions 1 through 4 hold.
Conditions 1 and 2 tell us that the configurations of $u$ and $v$ coincide with one another for $(p-4) \cdot \frac{n-1}{p-2}$ disks evenly distributed across $p-4$ pegs. The remaining $2 \cdot \frac{n-1}{p-2}$ disks are left unfrozen by either $u$ or $v$; call the set of these disks $U$.
Conditions 3 and 4 ensure that we can choose a peg that is frozen by either $u$ or $v$, but not both, and place $\frac{n-1}{p-2}$ of the disks in $U$ onto this peg which yields a configuration shared by both $u$ and $v$.
\end{proof}

As a consequence of Lemma~\ref{lem:adj} we can describe how to traverse an edge from a pegset $u$ to a pegset $v$ in $I_p^n$ by freezing and unfreezing disks.
We place $\frac{n-1}{p-2}$ of the disks left unfrozen by $u$ onto the peg frozen by $v$ but left unfrozen by $u$. Then, we take the peg frozen by $u$ and left unfrozen by $v$ and unfreeze every disk on it.

The asymptotic lower bound on $\tw(H_p^n)$ will be derived from an asymptotic lower bound on $\tw(I_p^n)$.
To compute the treewidth of $I_p^n$ we first need to prove that it is vertex-transitive and compute its diameter.

\begin{lemma}
$I_p^n$ is vertex-transitive.
\label{lem:ipdtrans}
\end{lemma}
\begin{proof}
We define a family of automorphisms $\phi_{i,j}$ which swap the roles of $d_i$ and $d_j$ in some pegset. The lemma follows from the fact that we can transform a pegset $u$ to any other pegset $v$ by a sequence of swap operations.
For any pegset $u$ we define the image of $u$ under $\phi_{i,j}$ to be
\begin{equation*}
\phi_{i,j}(u)(k) = \begin{cases}
       u(i) & k = j \\
       u(j) & k = i \\
       u(k) & \text{otherwise} 
       \end{cases}.
\end{equation*}
Let $u$ and $v$ be adjacent pegsets in $I_p^n$. By $U_u$ and $U_v$ we denote the sets of disks left unfrozen by $u$ and $v$, respectively.
By Lemma~\ref{lem:adj} there are pegs $p_u$ and $p_v$ such that $u$ freezes disks onto $p_u$ but not $p_v$, and $v$ freezes disk onto $p_v$ but not $p_u$.
Further, traversing the edge from $u$ to $v$ is equivalent to placing $\frac{n-1}{p-2}$ disks from $U_u$ onto $p_k$ and treating the disks frozen to $p_u$ as unfrozen.
If $u$ and $v$ are adjacent then $\phi_{i,j}(u)$ and $\phi_{i,j}(v)$ are also adjacent, since swapping the labels of two disks does not affect the traversal process.
If $\phi_{i,j}(u)$ and $\phi_{i,j}(v)$ are adjacent then so are $u$ and $v$, by the above and the fact that $\phi_{i,j} = \phi_{i,j}^{-1}$.
\end{proof}

\begin{lemma}
\label{lem:ipddiam}
The diameter of $I_p^n$ is $\Theta(n)$.
\end{lemma}
\begin{proof}
Let $u$ and $v$ be pegsets in $I_p^n$. Let $k = \frac{n-1}{p-2}$. 
If $u$ and $v$ do not freeze the same set of pegs, we can, by freezing and unfreezing disks, walk along a path in $I_p^n$ of length depending only on $p$, to a configuration that does freeze the same set of pegs as $v$, and continue with the process below. Therefore, assume $u$ and $v$ do freeze the same set of pegs, and label this set of pegs in increasing order by index as $Q = \{q_1, q_2, \dots, q_{p-3}\}$. For $i = 1, \dots, p - 3$, let $U_i$ be the set of disks frozen on $q_i$ by $u$, and let $V_i$ be the set frozen on $q_i$ by $v$.

For all $i = 1, \dots, p-3$, we iteratively transform $u$ into $v$, one peg at a time. For a given peg $q_i$, the process for transforming $U_i$ into $V_i$ is as follows. There are three cases:
\begin{enumerate}
\item There exists some disk $d \in V_i \setminus U_i$ that is unfrozen by $u$,
\item There exists some disk $d \in V_i \setminus U_i$ that is frozen on some other peg by $u$,
\item or $U_i = V_i$.
\end{enumerate}

In case 1, unfreeze $q_i$, then freeze an arbitrary peg $q_l \notin Q$, to obtain a new pegset $w$ adjacent to the current pegset. Since each pegset leaves $\frac{n-1}{p-2}+1$ disks unfrozen, let the new pegset freeze onto $q_l$ all but one of the disks unfrozen by $u$. Let the omitted disk, $d$, be one in $V_i$ that is unfrozen by $u$. Choose some $d' \in U_i \setminus V_i$ (one exists since $|U_i| = |V_i|$), then unfreeze $q_l$; freeze onto $q_i$ the set $(U_i \setminus \{d'\}) \cup \{d\}$, to obtain a new adjacent pegset where $d'$ is replaced by $d$.

Repeat this process until case 1 no longer applies, i.e. until every remaining $d \in V_i \setminus U_i$ is frozen by $u$. Then (case 2) consider some such $d$. $u$ does not freeze $d$ on a peg $q_r$ to which this process has already been applied, since all such pegs now agree with $v$. Therefore, $u$ freezes $d$ on some peg $q_s$ to which the process has not yet been applied. Unfreeze $q_s$ and freeze an arbitrary unfrozen peg $q_l$ to obtain the next pegset in the process; when doing so, some unfrozen disk $d''$ remains unfrozen. Then again freeze $q_s$, but omit $d$ and instead freeze $d''$ onto $q_s$. $d$ is now unfrozen, and we proceed as in case 1. Repeat case 2 until case 3 applies.

Repeating the overall process for every peg gives a path from $u$ to $v$ of length $O(n)$.
\end{proof}

Lemmas \ref{lem:ipdtrans} and \ref{lem:ipddiam} allow us to apply the following lemma due to Babai and Szegedy to obtain a lower bound of $\Omega(\frac{1}{n}V(|I_p^n|))$ on the \emph{vertex expansion} of $I_p^n$.

\begin{definition}
The vertex expansion of a graph $G$ is equal to
$$\min_{S\subseteq V(G): 1 \leq |S| \leq \frac{1}{2}} \frac{|\partial S|}{|S|},$$
where $\partial S$ is the union of the neighborhoods, in $G \setminus S$, of vertices in $S$.
\end{definition}

\begin{lemma}[Babai and Szegedy \cite{Babai1992}] \label{lem:vertextrans}
Let $G$ be a vertex-transitive graph. Then the vertex expansion of $G$ is $\Omega(1/d)$, where $d$ is the diameter of $G$.
\label{lem:vexp}
\end{lemma}

\begin{lemma}
The treewidth of $I_p^n$ is $\Omega(\frac{1}{n}|V(I_p^n)|)$.
\label{lem:vexptw}
\end{lemma}
\begin{proof}
By applying Lemmas~\ref{lem:ipdtrans}, \ref{lem:ipddiam}, and \ref{lem:vexp}, along with the definition of vertex expansion, we have $|\partial S| = \Omega(\frac{|S|}{n})$ for all $S \subseteq V(I_p^n)$ with $0 \leq |S| \leq \frac{|V(I_p^n)|}{2}$, which implies that the size of any balanced vertex separator of $I_p^n$ is bounded from below by $\Omega(\frac{|V(I_p^n)|}{n})$.
It follows that the treewidth of $I_p^n$ is also bounded from below by $\Omega(\frac{|V(I_p^n)|}{n})$.
\end{proof}

We now count the number of pegsets in $V(I_p^n)$. 

\begin{lemma}\label{lem:numpegsets}
The number of regular pegsets in $I_p^n$ is $\Theta(n^{-(p-3)/2} \cdot (p - 2)^n)$.
\end{lemma}
\begin{proof}
There are ${p \choose p-3}$ ways to choose the frozen pegs.
Each pegset divides the disks into $p - 2$ sets of (almost) equal size and there are $\frac{n!}{(\frac{n}{p-2})!)^{p-2}}$ ways to choose the sets. 
This is because there are $n!$ ways to order the disks, but we only care about the ordering of the $p-2$ partitions of the disks, hence we divide by $\left(\frac{n}{p-2}!\right)^{p-2}$. (Asymptotically, we may assume $n \equiv 0 \mod {p - 2}$.)
In total there are ${p \choose p - 3} \cdot \frac{n!}{((\frac{n}{p-2})!)^{p-2}}$ pegsets.
Since $p$ is fixed, we apply Stirling's approximation to $\frac{n!}{((\frac{n}{p-2})!)^{p-2}}$ to obtain the result.
\end{proof}

By applying Lemmas \ref{lem:vexptw} and \ref{lem:numpegsets} we obtain the following corollary.

\begin{corollary}
$\tw(I_p^n) = \Omega(n^{-(p-1)/2} \cdot (p-2)^n)$.
\label{cor:ipd}
\end{corollary}

Next we show how to obtain a lower bound of $\tw(H_p^n)$ from Corollary~\ref{cor:ipd}.
Since we have a lower bound on the treewidth of $I_p^n$, Lemma~\ref{lem:twhaven} guarantees the existence of a haven of a useful order.
The idea behind Lemma~\ref{lem:hpdipd} is to take a haven of order $\Omega(n^{-(p-1)/2} \cdot (p-2)^n)$ in $I_p^n$ and modify it to create a haven of the same order in $H_p^n$.

\begin{lemma}
$\tw(H_p^n) = \Omega(\tw(I_p^n))$.
\label{lem:hpdipd}
\end{lemma}
\begin{proof}
Let $k = \tw(I_p^n) + 1 = \Omega(n^{-(p-1)/2} \cdot (p-2)^n)$. By Lemma~\ref{lem:twhaven}, $I_p^n$ has a haven of order $k$. Call this haven $\phi$. Recall that a haven describes an evasion strategy for a robber in a cops-and-robbers game. Intuitively, if a robber can evade the cops in $I_p^n$, the same robber can also evade the cops in $H_p^n$ by playing only on states that belong to pegsets of $I_p^n$ and by paying attention only to which of those pegsets are occupied by at least one cop. We formalize this strategy below by constructing a haven for $H_p^n$ from $\phi$. Because a cop moving in $H_p^n$ may simultaneously occupy a constant number of pegsets in $I_p^n$,
the order of the haven we construct is a constant factor smaller than that of~$\phi$. 

Every vertex in $I_p^n$ corresponds to a pegset; every pegset corresponds to a set of configurations in the Towers of Hanoi game. Each of these configurations corresponds to a vertex in $H_p^n$. Define the function $f:\mathcal{P}(V(H_p^n)) \rightarrow \mathcal{P}(V(I_p^n))$, where for $X \subseteq V(H_p^n)$, $f(X)$ is the set of vertices in $I_p^n$ whose corresponding pegsets contain configurations in $X$. Define $g:\mathcal{P}(V(I_p^n)) \rightarrow \mathcal{P}(V(H_p^n))$, such that for $X' \subseteq V(I_p^n)$, $g(X')$ is the set of all configurations belonging to pegsets in $X'$. Let $g(X') = \emptyset$ if $X' = \emptyset$.

Define $\psi:\mathcal{P}(V(H_p^n))\rightarrow \mathcal{P}(V(H_p^n))$, such that $\psi(X)$ is the connected component containing $g(\phi(f(X)))$. To show that $\psi$ is a haven, it suffices to show that:
\begin{enumerate}
\item for all $X \subseteq V(H_p^n)$, $\psi(X)$ is well-defined\textemdash i.e. $g(\phi(f(X)))$ is connected and nonempty whenever $\phi(f(X))$ is nonempty,
\item for $Z \subseteq V(H_p^n)$, $\psi(Z) \subseteq \psi(X)$ whenever $X \subseteq Z$, and
\item $|X| = \Omega(f(X))$.
\end{enumerate}

For (1), to see that $g(\phi(f(X)))$ is connected in $H_p^n \setminus X$, consider any pair of configurations $u, v \in g(\phi(f(X)))$. $u$ and $v$ belong to pegsets $a$ and $b$ (respectively) in $\phi(f(X))$. $a$ has a path $P$ to $b$ in $\phi(f(X))$, since $\phi(f(X))$ is connected. Every vertex (pegset) $w$ in this path corresponds to the set $W' = g(w) \subseteq g(\phi(f(X)))$ of all configurations belonging to the pegset $w$. $W' \cap X = \emptyset$, or else by the definition of $f$, $w$ would be in $f(X)$, contradicting the fact that $w \in \phi(f(X))$. Furthermore, $W'$ is connected, since it is isomorphic to $H_3^d$ (where $d < n$). Also, every edge $(w_1, w_2)$ in $P$ corresponds to a vertex $w' \in g(\phi(f(X)))$ belonging to $W_1' = g(w_1)$ and $W_2' = g(w_2)$. Therefore, $u$ has a path to $v$ in $H_p^n \setminus X$, obtained by traversing an $H_3^d$ copy $W'$ for every vertex $w \in P$, and moving between $H_3^d$ copies $W'$ and $W''$ that intersect at a vertex in $g(\phi(f(X)))$ for every edge in $P$.

For (2), if $X \subseteq Z \subseteq V(H_p^n)$, then $f(X) \subseteq f(Z)$. Since $\phi$ is a haven, $\phi(f(Z)) \subseteq \phi(f(X))$. Therefore, $g(\phi(f(Z))) \subseteq g(\phi(f(X)))$, and both $g(\phi(f(Z)))$ and $g(\phi(f(X)))$ are connected. If $\phi(f(Z)) = \emptyset$, then $\psi(X) = \emptyset$, and (2) is true. Therefore, suppose $\phi(f(Z)) \neq \emptyset$. Suppose for a contradiction that $\psi(Z) \not\subseteq \psi(X)$. Let $u$ be a vertex in $\psi(Z) \cap \psi(X)$ (this intersection is nontrivial since it includes $g(\phi(f(Z)))$), and let $v$ be a vertex in $\psi(Z) \setminus \psi(X)$. Suppose $(u, v) \in E(H_p^n)$. (Such a pair must exist because $\psi(Z)$ is connected.) Since $X \subseteq Z$,
$$u, v \in \psi(Z) \subseteq  V(H_p^n) \setminus Z \subseteq V(H_p^n) \setminus X.$$
However, since $v \notin \psi(X)$, this contradicts that $\psi(X)$ is a connected component in $H_p^n \setminus X$.

(3) follows from the fact that $|f(X)| \leq (p - 2)|X|$, since every vertex belongs to at most $p - 2$ regular pegsets.
\end{proof}

We have now proven the second theorem of the section.

\begin{theorem}
For any fixed $p \geq 4$, $\tw(H_p^n) = \Omega(n^{-(p-1)/2} \cdot (p-2)^n)$. 
\label{thm:maintw}
\end{theorem}

In Appendix~\ref{sec:p4} we prove a lower bound on the treewidth of four-peg Hanoi graphs that,
while still separated by a polynomial factor from the upper bound, is tighter than the one above.

\section{Conclusion}
Theorem~\ref{thm:maintw} and Theorem~\ref{thm:maintwp4}, together with Theorem~\ref{thm:twh3} and Theorem~\ref{thm:ubound}, give nearly tight asymptotic bounds on the number of states, in the adversarial version of the Towers of Hanoi game we proposed in the introduction, that the first player must forbid in order to ensure better than even odds of defeating the second player. This raises additional questions. First, suppose the first player forbids enough states to separate the graph in a balanced way, but the second player is fortunate enough to have starting and ending positions in the same connected component. What is the optimal strategy for the second player, and how many moves will this strategy take? Must this strategy be formulated in graph-theoretic terms, or is there an algorithm that consists of moving the disks in an intuitive way?

Theorem~\ref{thm:maintwp4} (in Appendix~\ref{sec:p4}) improves the lower bound of Theorem~\ref{thm:maintw} when $p=4$; one question would be to see whether the technique in the proof of Theorem~\ref{thm:maintwp4} could be adapted to deal more generally with the structure of pegset intersection graphs when $p \geq 4$, yielding a bound of $\Omega(\frac{(p-2)^n}{n})$ in general when $p \geq 4$. However, this still would not eliminate the asymptotic gap between our upper and lower bounds.

To this end, Corollary~\ref{cor:kntw} gives a lower bound on the treewidth of the Kneser graph that is new when $2k + 1 \leq n \leq 3k - 1$. Can this lower bound be tightened? Harvey and Wood \cite{Harvey2014} showed that in this case (when $2k + 1 \leq n \leq 3k - 1$),
$$\tw(\kn(n, k)) < {n - 1 \choose k} - 1.$$
Since ${n - 1 \choose k} = \Theta({n \choose k})$ when $2k + 1 \leq n \leq 3k - 1$, this upper bound does not imply sublinear treewidth. However, if in fact the treewidth is sublinear, and can be used to obtain a sublinear vertex separator in $\dsg(n)$ (defined in Appendix~\ref{sec:p4}), then combined with a proof of asymptotic tightness in Lemma~\ref{lem:twdsgg4n} and Lemma~\ref{lem:twg4nhanoi}, this would imply that $\tw(H_4^n)$ is $o(2^n)$, proving that the upper bound in Theorem~\ref{thm:ubound} is not tight. This would be surprising, as the family of separators given in Theorem~\ref{thm:ubound} seems intuitively to target the ``weakest'' parts of the graph.

Another possible line of further research is whether the bound given in Lemma~\ref{lem:twtensor} is tight for the tensor product, and what can be said about other graph products. As stated in the introduction, Kozawa et al. \cite{kozawa} gave lower bounds for the Cartesian and strong products. Since the strong product of a graph has the same vertices as and a superset of the edges of the tensor product, our lower bound in Lemma~\ref{lem:twtensor} for the tensor product's treewidth immediately gives a lower bound on the treewidth of the strong product. However, Kozawa et al. \cite{kozawa} gave a stronger lower bound for the strong product. One question would be whether a comparable improvement over our bound can be proven for the tensor product.

\bibliography{paper}

\appendix
\newpage

\section{Four pegs}
\label{sec:p4}
Theorems~\ref{thm:ubound} and~\ref{thm:maintw}, together, give upper and lower bounds that differ by a polynomial factor in the number of disks of the Towers of Hanoi puzzle. Compared to the overall exponential size of the bound, this is a small gap, and it is tempting to try to close it further. The proof of Theorem~\ref{thm:maintw} identifies the pegset intersection graph ($I_p^n$) as a hard part of the graph to separate, and leverages the vertex-transitive structure of this graph.

However, there are many configurations in the game that are ignored by focusing on the $I_p^n$ graph: namely, all configurations where the numbers of disks on the pegs are arbitrary, i.e., not constrained to be equal to $\lfloor \frac{n}{p-2}\rfloor$ for $p - 3$ of the pegs. We broaden our analysis of pegsets to prove the main result of this section:
\begin{theorem}
$\tw(H_4^n) = \Omega(\frac{2^n}{n})$. 
\label{thm:maintwp4}
\end{theorem}

We begin by generalizing the pegset intersection graph beyond regular pegsets.

\begin{definition}
Let $G_4^n$ be a graph whose vertices are the pegsets of $H_4^n$ that freeze only one peg and that freeze
at most $\lfloor\frac{n-1}{2}\rfloor$ disks onto that peg. In this graph, let vertices $u$ and $v$ be adjacent whenever
the pegsets $u$ and $v$ freeze mutually disjoint sets of disks, and freeze them onto separate pegs. 
\end{definition}
Clearly $I_4^n$ is an induced subgraph of $G_4^n$. We prove our improved bound by analyzing the relationship between $G_4^n$ and the \emph{Kneser graph}.

\begin{definition}[Lovasz \cite{Lovasz1978}]
Let $[n] = \{1, \dots, n\}$ be an indexing of the objects in an arbitrary set. The Kneser graph, denoted $\kn(n, k)$, is the graph whose vertices correspond to the $k$-element subsets of $[n]$, and whose edges are the pairs of vertices whose corresponding subsets are disjoint.
\end{definition}
We restrict our attention to Kneser graphs that are connected, namely the graphs $\kn(n, k)$ where $n \geq 2k + 1$.

The condition on disjoint subsets in the definition of Kneser graphs is analogous to the condition on disjoint subsets of pegs in the definition of $G_4^n$. (In fact, for any given $k \leq \lfloor \frac{n-1}{2}\rfloor$, the pegsets that freeze exactly $k$ disks induce as a subgraph of $G_4^n$ the tensor product of $\kn(n, k)$ with a 4-clique\textemdash see Definition~\ref{def:tensor}.) However, $G_4^n$ also includes a separate condition, of having different frozen pegs. An additional complication is that $G_4^n$ allows sets of different sizes
rather than only considering sets of a single size $k$.
To account for all set sizes appropriately, we introduce a generalization of the Kneser graph:

\begin{definition}
Let the disjoint subset graph, denoted $\dsg(n, r)$, be the graph whose vertices are identified with the subsets $s \subseteq [n]$ with $|s| \leq r$, and whose edges are the pairs of vertices whose corresponding subsets are disjoint.
\end{definition}

For convenience, we let $\dsg(n) = \dsg(n, \frac{n-1}{2})$. Clearly $|V(\dsg(n))| \approx 2^{n-1}$.
Then $V(G_4^n)$ consists of four copies of $V(\dsg(n))$, with pegsets $u$ and $v$ connected iff they are in different copies and they share an edge in $\dsg(n)$.
In Lemma~\ref{lem:twdsg} we bound the treewidth of $\dsg(n)$, after which we will use the relationship between $G_4^n$ and $\dsg(n)$ to prove Theorem~\ref{thm:maintwp4}.
\begin{lemma}
$\tw(\dsg(n)) = \Omega(\frac{2^n}{n})$.
\label{lem:twdsg}
\end{lemma}

We defer the formal proof of \autoref{lem:twdsg} to later but outline a proof sketch below.
The idea of the proof is to observe that $\dsg(n)$ consists of $\frac{n-1}{2}$ Kneser graph ``slices.'' We make observations analogous to those leading to Corollary~\ref{cor:ipd}: Kneser graphs are vertex-transitive (Remark~\ref{rmk:knvt}) and have diameter $O(n)$ (Lemma~\ref{lem:kndiam}), implying that for all $0 \leq k \leq \frac{n-1}{2}$, $\tw(\kn(n, k)) = \Omega(\frac{1}{n}|V(\kn(n, k))|)$ (Corollary~\ref{cor:kntw}). Since
$$|V(\dsg(n))| = \sum_{k=0}^\frac{n-1}{2} |V(\kn(n, k))|,$$
Lemma~\ref{lem:twdsg} then follows if we can, intuitively, show that the Kneser slices are hard to separate from one another. We formalize this notion and show that it is true for most of the slices. The argument relies on the subset definitions of the Kneser graphs' vertices, and makes use of the \emph{Kruskal--Katona Theorem} (Corollary~\ref{cor:kkt}).

We prove that given a balanced vertex separator $X$ for $\dsg(n)$, either:
\begin{enumerate}
\item $X$ contains a large number of the vertices in $\dsg(n)$ (at least an $\Omega(\frac{1}{n})$ factor), or
\item after removing $X$ from $\dsg(n)$, there is still a large connected component in $\dsg(n)$, leading to a contradiction.
\end{enumerate}

In the second case, we derive the contradiction as follows: we observe that after removing $X$ from $\dsg(n)$, if case (1) does not hold, then most of the vertices of $\dsg(n)$ lie in Kneser slices that have large connected components, since their intersection with $X$ contains too few vertices for a balanced separator. Call this set of Kneser slices $K_{conn}(X)$. We prove that every pair of subgraphs $G_k = \kn(n, k)$ and $G_l = \kn(n, l)$ in $K_{conn}(X)$ have large connected components $A_k$ and $A_l$ that share an edge. Therefore, these large connected components, together, form a large connected component in $\dsg(n) \setminus X$, from which we derive the desired contradiction.

Finally, we use our lower bound on the treewidth of $\dsg(n)$ to derive a lower bound on the treewidth of $G_4^n$, and in turn on the treewidth of $H_4^n$. We obtain the former by proving a more general claim about the treewidth of the tensor product of two graphs, and the latter by a proof analogous to that of Lemma~\ref{lem:hpdipd}.

We begin by showing the required lower bound on the treewidth of the Kneser graph. We use the following result of Valencia-Pabon and Vera:

\begin{lemma}[Valencia-Pabon and Vera~\cite{ValenciaPabon2005OnTD}]
If $1 \leq k \leq \lfloor\frac{n-1}{2}\rfloor$, then the diameter of $\kn(n, k)$ is $\lceil\frac{k-1}{n-2k}\rceil + 1.$
\label{lem:kndiam}
\end{lemma}

\begin{remark}
When $k \leq \frac{n-1}{2}$, the diameter in Lemma~\ref{lem:kndiam} is $O(n)$.
\label{rmk:kndiamon}
\end{remark}

The following fact about Kneser graphs is well known; it also follows from a straightforward adaptation of the proof of Lemma~\ref{lem:ipdtrans}.

\begin{remark}
All Kneser graphs are vertex-transitive.
\label{rmk:knvt}
\end{remark}

Combining Lemmas~\ref{lem:vexp} and~\ref{lem:kndiam} with Remarks~\ref{rmk:kndiamon} and~\ref{rmk:knvt}, and observing the relationship between vertex expansion, balanced separators, and treewidth (as we did in the proof of Lemma~\ref{lem:vexptw}), gives the following corollary:
\begin{corollary}
For all $k$ such that $1\leq k \leq \frac{n-1}{2}$, $\tw(\kn(n, k)) = \Omega(\frac{1}{n}|V(\kn(n, k))|)$, and for every constant $c$, the minimum size of a $c$-separator in $\kn(n, k)$ is $\Omega(\frac{1}{n}|V(\kn(n, k))|)$.
\label{cor:kntw}
\end{corollary}

Before turning to the interfaces between the Kneser slices, we establish a threshold value such that most of the vertices of $\dsg(n)$ lie in $\kn(n, k)$ slices with values of $k$ exceeding this threshold. Restricting our attention (in Lemma~\ref{lem:klcc}) to these slices will allow us to prove the mutual connectedness of the large connected components in case (2).

\begin{restatable}{lemma}{rootntotal}
For every constant $\beta$ with $\frac{1}{2} < \beta < 1$, there exists a constant $\varepsilon$ such that $$\lim_{n\rightarrow \infty}\frac{\sum_{k=\frac{n}{2}-\varepsilon\sqrt{n}}^{\frac{n}{2}}|V(\kn(n, k))|}{|V(\dsg(n))|} \geq \beta.$$
\label{lem:rootntotal}
\end{restatable}
\begin{proof}
Let $B(n, p)$ denote the binomial distribution parameterized with probability $p$. The standard deviation of $B(n, \frac{1}{2})$ is $\frac{\sqrt{n}}{2}$. If $f$ is the probability mass function of $B(n, \frac{1}{2})$, then $f(k) = \frac{1}{2^n}{n \choose k}$.

Let $X$ be a random variable distributed according to $B(n, p)$.

By Chebyshev's inequality,
$$Pr[|X - \frac{n}{2}| \geq \varepsilon\sqrt{n}] \leq \frac{1}{4\varepsilon^2}.$$

Setting $\varepsilon = \frac{1}{2\sqrt{1-\beta}}$, so that $\beta = 1 - \frac{1}{4\varepsilon^2}$, yields the desired result, since
$$\sum_{k=\frac{n}{2}-\varepsilon\sqrt{n}}^{\frac{n}{2}}{n \choose k} = \frac{1}{2}\sum_{k=\frac{n}{2}-\varepsilon\sqrt{n}}^{\frac{n}{2}+\varepsilon\sqrt{n}}{n \choose k} = 2^{n-1}\cdot Pr[|X - \frac{n}{2}| \leq \varepsilon\sqrt{n}] \geq 2^{n-1}(1 - \frac{1}{4\varepsilon^2}),$$
and since $|V(\dsg(n))| = 2^{n-1}$.
\end{proof}

In Lemma~\ref{lem:bigcc} we will prove the existence of the large connected component from which the contradiction is derived in case (2) of the discussion following the statement of Lemma~\ref{lem:twdsg}. To do so, we will use the following lemma:
\begin{restatable}{lemma}{lemklcc}
\label{lem:klcc}
Let $\varepsilon > 0$ be fixed. Suppose $\frac{n-1}{2} - \varepsilon\sqrt{n} \leq l < k \leq \frac{n-1}{2}$, and let $A_k$ and $A_l$ be subsets, respectively, of the vertices in the $\kn(n, k)$ and $\kn(n, l)$ subgraphs of $\dsg(n)$. Suppose further that $|A_k| \geq d|V(\kn(n, k))|$ and $|A_l| \geq d|V(\kn(n, l))|$, where $d > \frac{1}{2}$ is a constant. Then, if $n$ is sufficiently large, $A_k$ and $A_l$ share an edge.
\end{restatable}

The proof of Lemma~\ref{lem:klcc} uses the \emph{Kruskal--Katona Theorem} (Corollary~\ref{cor:kkt}), which provides a lower bound, given a collection $\mathcal{F}$ of $k$-element subsets of $[n]$, on the number of $l$-element subsets of $[n]$ that are subsets of sets in $\mathcal{F}$. The following formulation of the Kruskal-Katona theorem is due to Lov{\'a}sz (Frankl gave a short proof):
\begin{theorem}[Kruskal--Katona Theorem~\cite{Frankl1984},\cite{lovasz1993}]
\label{thm:kktk1}
Let $\mathcal{F}$ be a family of $k$-element subsets of $[n]$, and let $\mathcal{E}$ be the set of all $k-1$-element subsets of sets in $\mathcal{F}$. Then whenever $|\mathcal{F}| \geq {m \choose k}$, $|\mathcal{E}| \geq {m \choose k - 1}$.
\end{theorem}
Applying induction on $l = k - 1, \dots, 1$ to Theorem~\ref{thm:kktk1} implies the following corollary:
\begin{corollary}
\label{cor:kkt}
Let $\mathcal{F}$ be a family of $k$-element subsets of $[n]$, and let $\mathcal{E}$ be the set of all $l$-element subsets of sets in $\mathcal{F}$, where $1 \leq l < k$. Then whenever $|\mathcal{F}| \geq {m \choose k}$, $|\mathcal{E}| \geq {m \choose l}$.
\end{corollary}

Using Corollary~\ref{cor:kkt}, we prove Lemma~\ref{lem:klcc}:
\begin{proof}{\emph{(Proof of Lemma~\ref{lem:klcc})}}
For every $v \in V(\dsg(n))$, view $v$ as the $k$-size subset with which it is identified, and let $\overline{v}$ be the set complement of $v$.

Let $B_k = \{\overline{v} \mid v \in A_k\}$. Define a function $\delta_l$ mapping vertices in $\kn(n, k)$ to their neighborhoods in $\kn(n, l)$: for all $v \in V(\kn(n, k))$, let $\delta_l(v) = \{w \in V(\kn(n, l)) \mid (v, w) \in E(\dsg(n))\}$.

Extend the domain of $\delta_l$ to sets of vertices in $\kn(n, k)$: for all $Z \subseteq V(\kn(n, k))$, let $\delta_l(Z) = \bigcup_{v \in Z}\delta_l(v)$.

Clearly, a vertex $u \in \kn(n, l)$ is in $\delta_l(A_k)$ iff there exists some $w \in B_k$ such that, viewing the vertices in their combinatorial sense, $u \subseteq w$.

I.e., $\delta_l(A_k)$ consists precisely of the vertices that are identified with subsets of vertices in $B_k$. Since
$$|B_k| = |A_k| > \frac{1}{2}|\kn(n, k)| = \frac{1}{2}{n \choose k} \geq {n - 1 \choose n - k},$$ Corollary~\ref{cor:kkt} implies that
$$|\delta_l(A_k)| \geq {n - 1 \choose l} \geq (\frac{1}{2} - o(1)){n \choose l} - 1 \geq (\frac{1}{2} - o(1))|V(\kn(n, l))| - 1.$$
In the above inequalities we use the (easily verified) fact that whenever $\frac{n-1}{2} - \varepsilon\sqrt{n} \leq i \leq \frac{n-1}{2}$,
$${n - 1 \choose i} \geq (\frac{1}{2} - o(1)){n \choose i} - 1,$$
and
$${n - 1 \choose n - i} \leq \frac{1}{2}{n \choose i}.$$

Since by assumption $|A_l| \geq d|V(\kn(n, l))|$ with $d > \frac{1}{2}$, this implies that for sufficiently large $n$, $\delta_l(A_k) \cap A_l \neq \emptyset$. That is, some vertex in $A_k$ shares an edge with some vertex in $A_l$. 
\end{proof}

We are ready to formalize case (2) (Lemma~\ref{lem:bigcc}) in the discussion following the statement of Lemma~\ref{lem:twdsg}.

\begin{lemma}
\label{lem:bigcc}
Let $X$ be a vertex separator for $\dsg(n)$. Let $\frac{1}{2} < c < 1$ and $\varepsilon > 0$ be constants.
Let
$$K_{big} = \{\kn(n, k) | \frac{n-1}{2} - \varepsilon\sqrt{n} \leq k \leq \frac{n-1}{2}\}$$ be the largest $\varepsilon\sqrt{n}$ Kneser subgraphs of $\dsg(n)$.
Let $$K_{conn}(X) = \{\kn(n, k) \in K_{big} | \frac{|X \cap V(\kn(n, k))|}{|V(\kn(n, k))|} < f(n)\},$$
where $f$ is any function such that $f(n) = O(\frac{1}{n})$.

Then if $n$ is sufficiently large, for all $\kn(n, k) \in K_{conn}(X)$, $\kn(n, k) \setminus X$ has a connected component $A_k$ of size at least $c(1 - O(\frac{1}{n}))|V(\kn(n, k))|$, and for all $l \neq k$, if $\kn(n, l) \in K_{conn}(X)$, then $A_k$ and $A_l$ share an edge.

\end{lemma}
\begin{proof}
By Corollary~\ref{cor:kntw}, for all $\kn(n, k) \in K_{conn}(X)$, the minimum $c$-separator size for $\kn(n, k)$ is $\Omega(\frac{1}{n}|V(\kn(n, k))|)$, which by assumption is more than the vertices of $X$ that lie in $\kn(n, k)$ \textemdash at least when $n$ is sufficiently large. This implies that $A_k$ is of the stated size. For the second part of the claim, consider any $A_k, A_l$ pair. $A_k$ and $A_l$ are connected by Lemma~\ref{lem:klcc}, since $c(1 - O(\frac{1}{n})) \geq d$ for every constant $d$ such that $c \geq d > \frac{1}{2}$. The lemma follows.
\end{proof}

We are now ready to prove Lemma~\ref{lem:twdsg}. We choose numerical values instead of symbols for some of the constants that appear in the proof to make the argument more intuitive, although there are other values that work. 
\begin{proof}{\emph{(Proof of Lemma~\ref{lem:twdsg})}}
Choose any constant $\frac{1}{2} < c < \frac{4}{7}$. Let $X$ be a $c$-separator for $\dsg(n)$. 

We will show that either $X$ contains many vertices from large Kneser slices (those in $K_{sep}(X)$, which we define below), or most (more than a factor of $c$) of the vertices of $\dsg(n) \setminus X$ lie in a large connected component, so that $X$ is not a $c$-separator.

Let $K_{big}$ be the set of $\kn(n, k)$ subgraphs with $\frac{n-1}{2} - \varepsilon\sqrt{n} \leq k \leq \frac{n-1}{2}$, where $\varepsilon$ is chosen so that $\frac{|V(K_{big})|}{|V(\dsg(n))|} \geq \frac{8}{9}$. (We choose $\frac{8}{9}$ to make the argument work for $c < \frac{4}{7}$.) Let 
$$K_{sep}(X) = \{\kn(n, k) \in K_{big} | \frac{|X \cap \kn(n, k)|}{|\kn(n, k)|} \geq f(n)\},$$
where $f(n) = \Theta(\frac{1}{n})$ is the lower bound given by Corollary~\ref{cor:kntw} on the minimum $\frac{5}{7}$-separator size for $\kn(n, k) \in K_{big}$. (We choose $\frac{5}{7}$ because it produces the desired result for $c < \frac{4}{7}$.)

Let $K_{conn}(X) = K_{big} \setminus K_{sep}(X)$. There are two cases:
\begin{enumerate}
\item $\frac{|V(K_{sep}(X))|}{|V(K_{big})|} \geq \frac{1}{10}$.
\item $\frac{|V(K_{conn}(X))|}{|V(K_{big})|} > \frac{9}{10}$.
\end{enumerate}

(We choose $\frac{1}{10}$ and $\frac{9}{10}$, again to make the argument work for $c < \frac{4}{7}$.)

In case 1, since $K_{sep}(X)$ is defined so that $\frac{|X \cap V(K_{sep}(X))|}{|V(K_{sep}(X))|} \geq f(n)$,
$$\frac{|X \cap K_{sep}(X)|}{|V(\dsg(n))|} \geq \frac{|V(K_{sep}(X))|}{|V(K_{big})|}\cdot \frac{|V(K_{big})|}{|V(\dsg(n))|}\cdot f(n) \geq \frac{1}{10}\cdot \frac{8}{9}\cdot f(n) = \Omega(f(n)) = \Omega(\frac{1}{n}).$$

In this case we are done.

In case 2, Lemma~\ref{lem:bigcc} implies that there exists a connected component $A_k$ in every $\kn(n, k) \subseteq K_{conn}(X)$ of size at least $(\frac{5}{7} - O(\frac{1}{n}))|V(\kn(n, k))|$, and that every pair $A_k$ and $A_l$ are mutually connected. This implies that $\dsg(n) \setminus X$ has a connected component $A$ such that
$$\frac{|V(A)|}{|V(\dsg(n))|} \geq (\frac{5}{7} - O(\frac{1}{n}))\frac{|V(K_{conn}(X))|}{|V(\dsg(n))|} \geq (\frac{5}{7} - O(\frac{1}{n}))(\frac{9}{10})(\frac{|V(K_{big})|}{|V(\dsg(n))|}$$
$$\geq (\frac{5}{7} - O(\frac{1}{n}))(\frac{9}{10})(\frac{8}{9}) > \frac{4}{7} > c.$$
This contradicts the assumption that $X$ is a $c$-separator for $\dsg(n)$.
\end{proof}

To show that $\tw(H_4^n) = \Omega(\tw(\dsg(n))$, we first show that the treewidth of the generalized pegset intersection graph $G_4^n$ defined earlier is at least that of $\dsg(n)$, then that $\tw(H_4^n) = \Omega(\tw(G_4^n))$. Both of these are accomplished via haven mappings (Lemmas~\ref{lem:twdsgg4n} and \ref{lem:twg4nhanoi}) of a similar flavor to Lemma~\ref{lem:hpdipd}.

\begin{restatable}{lemma}{lemtwdsggn}
$\tw(G_4^n) = \Omega(\tw(\dsg(n)))$.
\label{lem:twdsgg4n}
\end{restatable}
We prove Lemma~\ref{lem:twdsgg4n} as a special case of a more general claim, Lemma~\ref{lem:twtensor}, about the treewidth of the \emph{tensor product} of graphs:

\begin{definition}
The tensor product of graphs $G$ and $H$, denoted $G \times H$, is the graph whose vertex set is the Cartesian product $V(G) \times V(H)$, and whose edges are the pairs of $(u_1, v_1)$ and $(u_2, v_2)$ whose first and second components share edges in $E(G)$ and $E(H)$ respectively, i.e.
$$\{((u_1, v_1), (u_2, v_2)) \mid u_1, u_2 \in V(G), v_1, v_2 \in V(H), (u_1, u_2) \in E(G), (v_1, v_2) \in E(G)\}.$$
\label{def:tensor}
\end{definition}

We prove the following:
\begin{lemma}
Let $G$ and $H$ be connected graphs, and suppose that $H$ is not bipartite. Then 
$$\tw(G \times H) \geq \tw(G).$$
\label{lem:twtensor}
\end{lemma}

To prove Lemma~\ref{lem:twtensor}, we first define an association between the vertices of $G$ and those of $J = G \times H$.

\begin{definition}
Given the tensor product $J = G \times H$ of graphs $G$ and $H$, define $f:V(J) \rightarrow V(G)$ so that for all $(u, v) \in V(J)$,
$$f((u, v)) = u.$$
Define $g:V(G) \rightarrow \mathcal{P}(V(J))$ so that for all $u \in V(G)$,
$$g(u) = f^{-1}(u) = \{(u, v) \mid v \in V(H)\}.$$
\label{def:fg}
\end{definition}

We use this definition to prove Lemma~\ref{lem:twtensor}. The proof is similar in spirit to the proof of Lemma~\ref{lem:hpdipd}:
\begin{proof}{\emph{(Proof of Lemma~\ref{lem:twtensor})}}
For the lower bound $\tw(G \times H) \geq \tw(G)$, by Lemma~\ref{lem:twhaven}, $G$ has a haven $\phi$ of order $k = \tw(G) + 1$. We construct a haven $\psi$ in $J = G \times H$ of order $k' \geq k$, from which the lemma follows. To define $\psi$, we extend the domains of $f$ and $g$ to sets of vertices in the natural way. That is, for every $X \subseteq V(J)$, let $f(X)$ be the image of all vertices in $X$ under $f$. For every $Y \in V(G)$, let $g(Y)$ be the union of the images under $g$ of all vertices in $Y$.

For all $X \subseteq V(J)$, let $\psi(X)$ be the connected component in $J \setminus X$ containing $g(\phi(f(X)))$. (Let $\psi(\emptyset) = \emptyset$.) It suffices to show that:
\begin{enumerate}
\item $Y' = g(Y)$ is a nonempty connected component in $J \setminus X$ whenever $Y$ is a nonempty connected component in $G \setminus f(X)$,
\item for all $X \subseteq Z \subseteq V(J)$, $\psi(Z) \subseteq \psi(X)$, and
\item for all $X \subseteq V(J)$, $|f(X)| \leq |X|$.
\end{enumerate}

For (1), suppose $Y$ is a connected component in $G \setminus f(X)$ for some $X \subseteq V(J)$. Let $Y' = g(Y)$. If $|Y| > 1$, then consider any edge $(u, w) \in Y$. Then for every pair of vertices $v, x \in V(H)$, the vertices $(u, v)$ and $(w, x)$ are connected by a path $P'$ in $Y'$. To construct this path, consider any walk $P$ along a sequence of vertices $(v, z_1, z_2, \dots, z_l, x)$ of odd length in $H$ from $v$ to $x$. Such a walk must exist since $H$ is not bipartite, i.e. contains an odd cycle. Construct the corresponding path $P'$ in $Y'$ by alternating between copies of $u$ and copies of $w$. That is, let 
$$P' = ((u, v), (w, z_1), (u, z_2), (w, z_3), \dots, (w, z_{l-1}), (u, z_l), (w, x)).$$

Since such a path exists for every edge $(u, v) \in Y$, and $Y$ is connected, $Y'$ is also connected.

We deal with the degenerate case $|Y| = 1$ by letting $Y'$ be a single copy $(u, v)$ of the vertex $u \in Y$, and obtain $\psi$ by extending this copy to a connected component.

For (2), it follows from the definition of $f$ and the fact that $\phi$ is a haven, that $\phi(f(Z)) \subseteq \phi(f(X))$. $\psi$ merely extends $\phi(f(Z))$ and $\phi(f(X))$ to connected components in $J \setminus Z$ and $J \setminus X$ respectively. The connected component $B$ in $J \setminus X$ containing $\phi(f(Z))$ is the same as the connected component in $J \setminus X$ containing $\phi(f(X))$, since both $\phi(f(Z))$ and $\phi(f(X))$ are connected and one is a subset of the other. Furthermore, since $X \subseteq Z$, $J \setminus Z \subseteq J \setminus X$, so removing the additional vertices in $Z \setminus X$ from $B$ cannot result in a connected component with vertices missing from $B$. That is, the connected component $\psi(Z)$ in $J \setminus Z$ containing $\phi(f(Z))$ is a subset of the connected component $\psi(X)$ in $J \setminus X$ containing $\phi(f(X))$.

(3) is immediate from the definition of $f$.
\end{proof}

Lemma~\ref{lem:twdsgg4n} immediately follows from Lemma~\ref{lem:twtensor} and the fact that $G_4^n$ is isomorphic to $\dsg(n) \times K_4$.

\begin{lemma}
\label{lem:twg4nhanoi}
$\tw(H_4^n) = \Omega(\tw(G_4^n))$.
\end{lemma}
\begin{proof}
Construct a haven mapping analogous to the mapping in Lemma~\ref{lem:hpdipd}. In Lemma~\ref{lem:hpdipd} we defined $f$ and $g$ as, respectively, mapping sets of configurations to the regular pegsets to which they belong, and mapping sets of regular pegsets to the unions of their configurations. Extend the codomain of $f$ and the domain of $g$, beyond regular pegsets, to the set of all pegsets in $G_4^n$. The rest of the argument is similar to the proof of Lemma~\ref{lem:hpdipd}. Again we need to check the following conditions:

\begin{enumerate}
\item for all $X \subseteq V(H_p^n)$, $\psi(X)$ is well-defined\textemdash i.e. $g(\phi(f(X)))$ is connected and nonempty whenever $\phi(f(X))$ is nonempty,
\item for $Z \subseteq V(H_p^n)$, $\psi(Z) \subseteq \psi(X)$ whenever $X \subseteq Z$, and
\item $|X| = \Omega(f(X))$.
\end{enumerate}

(3) is easy since every configuration belongs to at most four pegsets. The reasoning for (1) is identical to that in the proof of Lemma~\ref{lem:hpdipd}. For (2), the reasoning is also the same.
\end{proof}

Theorem~\ref{thm:maintwp4} follows from Lemma~\ref{lem:twdsg}, Lemma~\ref{lem:twdsgg4n}, and Lemma~\ref{lem:twg4nhanoi}.
\end{document}